\documentclass[sigconf]{acmart}
\renewcommand\footnotetextcopyrightpermission[1]{} 
\usepackage{graphicx}
\usepackage{balance}  
\usepackage[show]{chato-notes}
\usepackage{xspace}
\usepackage{booktabs}
\usepackage{multirow}
\usepackage[ruled,lined,linesnumbered,noend]{algorithm2e}
\usepackage[]{algpseudocode}
\usepackage{amsmath,amssymb,amsthm, mathrsfs}
\usepackage{subfig}
\usepackage{enumitem}
\setlist[itemize]{leftmargin=*}
\usepackage{tikz,pgfplots,pgfplotstable}
\usepackage{tablefootnote}
\usepackage[export]{adjustbox}

\usetikzlibrary{matrix,positioning,fit,shapes,arrows,shadows,calc}
\newtheorem{problem}{Problem}

\newcommand{\NP}{\ensuremath{\mathbf{NP}}\xspace}
\newcommand{\APX}{\ensuremath{\mathbf{APX}}\xspace}
\newcommand{\Pclass}{\ensuremath{\mathbf{P}}\xspace}
\newcommand{\NPhard}{{\NP}-hard\xspace}
\newcommand{\NPcomplete}{{\NP}-complete\xspace}
\newcommand{\bigO}{\ensuremath{\mathcal{O}}\xspace}

\newcommand{\msic}{\textsc{MSIC[\ensuremath{\alpha, \beta}]}\xspace}
\newcommand{\kmsic}{\ensuremath{k}-\textsc{MSIC[\ensuremath{\alpha, \beta}]}\xspace}
\newcommand{\mmscp}{\textsc{MMSCP}\xspace}

\newcommand{\mmc}{\textsc{MMC}\xspace}
\newcommand{\localSC}{\textsc{Local-SCS}\xspace}

\newcommand{\lenmax}{\ensuremath{l_{max}}\xspace}

\newcommand{\spara}[1]{\smallskip\noindent{\bf{#1}}}

\newcommand{\squishlist}{\begin{list}{$\bullet$}
  { \setlength{\itemsep}{0pt}
     \setlength{\parsep}{3pt}
     \setlength{\topsep}{3pt}
     \setlength{\partopsep}{0pt}
     \setlength{\leftmargin}{1.5em}
     \setlength{\labelwidth}{1em}
     \setlength{\labelsep}{0.5em} } }
\newcommand{\squishend}{
\end{list}  }

\def\BibTeX{{\rm B\kern-.05em{\sc i\kern-.025em b}\kern-.08emT\kern-.1667em\lower.7ex\hbox{E}\kern-.125emX}}

\tikzset{
    ncbar angle/.initial=90,
    ncbar/.style={
        to path=(\tikztostart)
        -- ($(\tikztostart)!#1!\pgfkeysvalueof{/tikz/ncbar angle}:(\tikztotarget)$)
        -- ($(\tikztotarget)!($(\tikztostart)!#1!\pgfkeysvalueof{/tikz/ncbar angle}:(\tikztotarget)$)!\pgfkeysvalueof{/tikz/ncbar angle}:(\tikztostart)$)
        -- (\tikztotarget)
    },
    ncbar/.default=0.5cm,
}

\tikzset{square left brace/.style={ncbar=0.2cm}}
\tikzset{square right brace/.style={ncbar=-0.2cm}}

\usepackage{balance}

\def\BibTeX{{\rm B\kern-.05em{\sc i\kern-.025em b}\kern-.08emT\kern-.1667em\lower.7ex\hbox{E}\kern-.125emX}}
    
\setcopyright{none}
\begin{document}

\fancyhead{}

\title{Discovering Interesting Cycles in Directed Graphs}

\author{Florian Adriaens}
\affiliation{%
\institution{IDLab, Ghent University}
}\email{florian.adriaens@ugent.be}
\author{Cigdem Aslay}
\affiliation{%
\institution{Aalto University}
}\email{cigdem.aslay@aalto.fi}
\author{Tijl De Bie}
\affiliation{%
\institution{IDLab, Ghent University}
}\email{tijl.debie@ugent.be}
\author{Aristides Gionis}
\affiliation{%
\institution{Aalto University}
}\email{aristides.gionis@aalto.fi}
\author{Jefrey Lijffijt}
\affiliation{%
\institution{IDLab, Ghent University}
}\email{jefrey.lijffijt@ugent.be}

%
\renewcommand{\shortauthors}{Adriaens, et al.}

%

\begin{abstract}
Cycles in graphs often signify interesting processes.
For example, cyclic trading patterns can indicate inefficiencies or economic dependencies in trade networks,
cycles in food webs can identify fragile dependencies in ecosystems,
and cycles in financial transaction networks can be an indication of money laundering.
Identifying such interesting cycles, which can also be constrained to contain a given set of query nodes, 
although not extensively studied, is thus a problem of considerable importance.
%
In this paper, we introduce the problem of discovering interesting cycles in graphs. We first address the problem of quantifying the extent to which a given cycle is interesting for a particular analyst.
We then show that finding cycles according to this interestingness measure is related to the longest cycle and maximum mean-weight cycle problems (in the unconstrained setting) 
and to the maximum Steiner cycle and maximum mean Steiner cycle problems (in the constrained setting).
A complexity analysis shows that finding interesting cycles is \NP-hard, and is \NP-hard to approximate within a constant factor in the unconstrained setting, and within a factor polynomial in the input size for the constrained setting. The latter inapproximability result implies a similar result for the maximum Steiner cycle and maximum mean Steiner cycle problems.
Motivated by these hardness results, we propose a number of efficient heuristic algorithms. We verify the effectiveness of the proposed methods and demonstrate their practical utility on two real-world use cases: a food web and an international trade-network dataset.
\end{abstract}

%


%

\maketitle

\section{Introduction}
\label{sec:intro}
Cycles occur as a natural data-mining pattern in several real-world applications.
They appear naturally in food webs, 
where cycles highlight cyclic dependencies, often revealing the fragile parts of an ecosystem~\cite{dunne2002network}.
In financial transaction data, a cycle could be an indication of a money-laundering scheme~\cite{colladon2017using}.
In biological and complex networks, a cycle is an indication of a feedback mechanism~\cite{kwon2007analysis}.
Despite the wide range of use cases, 
the problem of discovering cyclic patterns in graphs has not received much attention in the data-mining community (See Section~\ref{sec:related}).



\begin{figure}[tp]
\captionsetup[subfloat]{farskip=5pt,captionskip=1pt}
\centering
\subfloat[\label{fig:toya}]{\scalebox{0.25}{\includegraphics[width=1\linewidth]{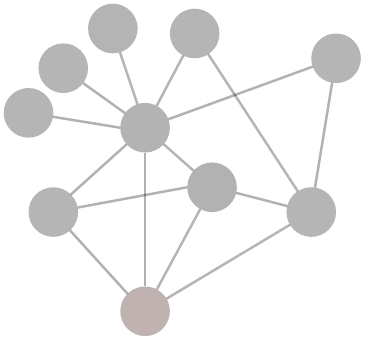}}}
\qquad
\subfloat[\label{fig:toyb}]{\scalebox{0.25}{\includegraphics[width=1\linewidth]{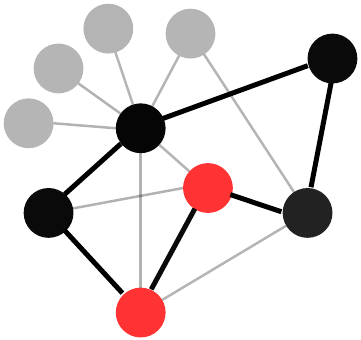}}}\\

\subfloat[\label{fig:toyc}]{\scalebox{0.25}{\includegraphics[width=1\linewidth]{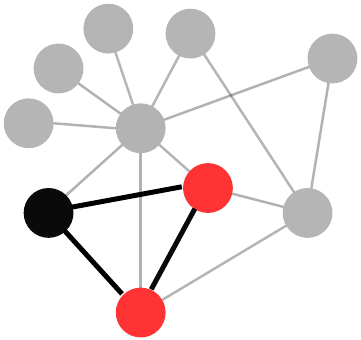}}}
\qquad
\subfloat[\label{fig:toyd}]{\scalebox{0.25}{\includegraphics[width=1\linewidth]{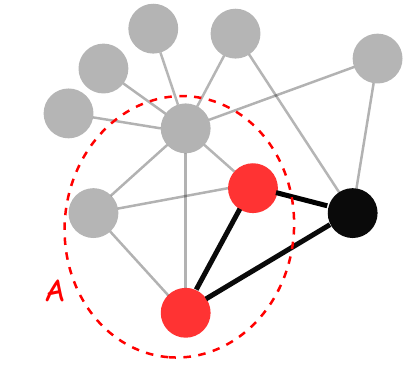}}}
\caption{The most interesting Steiner cycles connecting the red nodes according to different prior beliefs on the graph shown in (a): when we have (b) no knowledge about the graph, (c) knowledge about the individual degrees, and (d) knowledge about the degrees and the density of the community $\mathcal{A}$.\label{fig:toyexample}}
\end{figure}

In this paper, we study the problem of discovering interesting cycles in a directed and non-negatively weighted graph.
We also consider the constrained case, where the cycles have to contain a set of user-specified query nodes. Cycles containing a given set of query nodes are called \emph{Steiner} cycles \citep{SALAZARSteinercycle, Ste10}.
Identifying interesting Steiner cycles can be particularly useful in different application domains.
For example, a biologist may be interested in finding a food chain that contains both a rabbit and a hawk 
to assess the importance of the hawk population for the rabbit population.
Economists may be interested in finding surprising trading action between certain countries 
in different parts of the world.

As networks typically contain numerous cycles, 
a key challenge is the choice of a suitable interestingness measure for a cycle.
We propose to use a \emph{subjective} measure, i.e., 
taking into account which network characteristics (if any) are known a priori to the analyst.
For example, for a lay person it might be surprising that more than 50\% of the Dominican Republic's export 
is to the USA.\footnote{\url{https://tradingeconomics.com/dominican-republic/exports-by-country}}
However, for an economist possessing the knowledge that those countries have a bilateral trade agreement (which can be formalized as prior information on the trade network), such a trade volume might not come as a surprise. Thus, we are interested in designing methods that are able to take such prior knowledge into account.

Based on this observation, our proposed measure is built on \emph{subjective interestingness} \citep{silberschatz1996,debie2011}.
In the formalization of subjective interestingness, a pattern is deemed interesting 
if it is both \emph{surprising} to the user and can be communicated in a \emph{concise} way.
The measure we propose to quantify interestingness of cycles in graphs is presented in Section~\ref{sec:si}.

Figure~\ref{fig:toyexample} illustrates an example of our setting. Figure~\ref{fig:toya} shows a toy graph in which a user wishes to find a cyclic pattern containing the red query nodes. We consider three different users.
The first user has no knowledge of the graph. In this case, with every revealed edge, the user learns something about the graph, hence, the most interesting cycle is the longest cycle containing the red nodes, as shown in Figure~\ref{fig:toyb}. The second user has knowledge about the degrees of each node in the network. In this case, edges containing high degree nodes are less interesting to this user as they are expected. This prior knowledge makes the cycle shown in Figure~\ref{fig:toyc} the most interesting cycle to the second user. Our last user is a specialist. Besides knowing the degrees of the nodes, he also has prior knowledge that the red nodes are part of a dense community $\mathcal{A}$. Intra-community edges are now expected and thus are less interesting to the third user. This makes the cycle obtained in Figure~\ref{fig:toyd} the most interesting cycle for the third user. 

Following the proposed cycle interestingness measure, 
in Section~\ref{sec:problem} we formally define the two problem variants
that we study in this paper: 
(i) the Maximum Subjectively Interesting Cycle problem (\msic); and 
(ii) the Maximum Subjectively Interesting Steiner Cycle problem (\kmsic), in which the cycle is required to contain a given set of $k$ terminal nodes.
We provide an extensive computational complexity analysis in Section~\ref{sec:theory} showing that both problems are \NP-hard, and are \NP-hard to approximate within a constant factor for \msic, and within a factor polynomial in the input size for \kmsic. 
In Section~\ref{sec:algorithms}, we present a number of efficient heuristics for both problems. We show the effectiveness of our methods in practical settings through an extensive experimental evaluation in Section~\ref{sec:experiments}. We also provide two real-world use cases, one regarding a food web and another regarding an international trade network, demonstrating the potential of the proposed methods for real-world applications. 

\spara{Contributions and roadmap.}
\squishlist
\item We present a novel subjective interestingness measure for cycle patterns in graphs (Section~\ref{sec:si}). 
\item We formally define the Maximum Subjectively Interesting Cycle and Maximum Subjectively Interesting Steiner Cycle problems, and provide an extensive theoretical analysis of their computational complexity (Section~\ref{sec:theory}). 
\item We propose a number of efficient and effective heuristics for both problems  (Section~\ref{sec:algorithms}). 
\item We experimentally verify the effectiveness of our methods and demonstrate their practical utility on two real-world use cases (Section~\ref{sec:experiments}). 
\squishend

%


\section{Cycles and their interestingness}
\label{sec:si}
In this section, we first introduce the notation used in this paper and formally define the notion of a cycle pattern in weighted digraphs (Section \ref{sec:pattern}). We then explain how the interestingness of a cycle pattern can be formalized w.r.t. a background distribution that models prior knowledge about the structure of the graph (Section~\ref{sec:measure}). For the sake of clarity and completeness, we also briefly summarize the related work on how such a background distribution can be derived based on a number of relevant types of prior knowledge on the graph structure (Section~\ref{sec:background}).


\subsection{Graph notation}\label{sec:notation}
We assume a simple digraph $G=(V,E)$, with $|V| = n$ nodes and $|E| = m$ directed edges. A \emph{walk} in $G$ is defined as a sequence $v_1, v_2, \ldots, v_k$ of nodes, where $(v_i,v_{i+1}) \in E$ for $i \in [1, k-1]$ and $(v_i, v_{i+1}) \not\eq (v_j, v_{j+1})$ for all $1 \le i <  j \leq k$.  We say that a walk is \emph{closed} if $v_1 = v_k$.  A (simple) \emph{cycle} is a closed walk $v_1, v_2, \ldots,  v_k = v_1$, with no repetition of the nodes $v_i$, for $1 < i < k$. We use $v \in C$ and $e \in C$ to indicate that a node $v$ and an edge $e$ is part of a cycle $C$, respectively. We use $|C|$ to denote the length of a cycle $C$, i.e. the number of edges it contains.

\subsection{Cycles as patterns}\label{sec:pattern}
The patterns considered in this paper consist of the specification of a cycle $C$ that is stated to be present in a given graph.  
Additionally, we communicate $|C|$ positive real values $\ell_e$, one for each edge in $e \in C$. Each value $\ell_e$ represents a lower bound\footnote{We use a lowerbound and not the actual edge value. Often a user will not care about the \emph{exact} weight of an edge, but rather only if the weight is high or not, relative to the user's expectations on the graph structure.} on the weight of edge $e$, 
thus informing the user that the weight is at least $\ell_e$.
In practice, in the most interesting cycle patterns, a lower bound $\ell_e$ will be equal (or as close as possible given the number encoding used) to the observed value of the weight $\mu(e)$, as a larger $\ell_e$ provides more information.

\subsection{Subjective interestingness of cycle patterns}\label{sec:measure}
We follow the approach proposed by \citet{debie2011} in formalizing the subjective interestingness of a cycle pattern as the ratio of its \emph{Information Content} (IC),
and its \emph{Description Length} (DL), which should reflect the amount of effort it takes the data analyst to assimilate the pattern.
Here, IC is the negative log probability of the pattern being present in the data, where the probability is computed w.r.t. a so-called \emph{background distribution} $P$ that represents the prior expectations of the analyst. The IC reflects the more improbable the analyst considers a given pattern, the more information it conveys when the analyst learns the pattern is present.

It may be impossible to accurately represent all expectations of an analyst.
Yet, it was argued that given a set of constraints in terms of expectations on certain statistics of the network (e.g., node degrees, subgraph densities, etc.), a robust estimate of the background distribution can be obtained by choosing $P$ as the maximum entropy distribution, subject to these constraints~\citep{debie2011}.

As reviewed in Section~\ref{sec:background}, a wide range of prior knowledge types have the (convenient) property that the resulting background distribution factorizes as a product of independent distributions, one for each possible edge $e \in V \times V$.
We emphasize that independence is not a choice, but rather a result of the maximum entropy optimization.
For unweighted and integer weighted networks, these are respectively Bernoulli and geometric distributions.
Hence, the IC of a cycle $C$ equals
\begin{equation*}
\text{IC}(C) = -\text{log}\left(\displaystyle\prod_{e \in C} \text{Pr}(\mu (e) \geq \ell_e)\right) = \displaystyle \sum_{e \in C}w(e),
\end{equation*}
where $w(e) \triangleq -\text{log}(\text{Pr}(\mu (e) \geq \ell_e))$ denotes the information content of the edge $e$, with $\text{Pr}(\cdot)$ denoting the probability under the background distribution $P$. Note that $w(e)\geq 0$ for all $e\in V\times V$.

The DL can be computed similarly as in the work of \citet{vanleeuwen2016}. 
To communicate a cycle pattern $C$ to the user, we need to communicate $|C|$ nodes.
We assume that the cost of assimilating that a node is part of $C$ is $\log(1/q)$, and that a node is not part of $C$ is $\log(1/(1-q))$.
Hence the DL of communicating $|C|$ nodes is equal to
\begin{align*}
|C|\cdot \log \dfrac{1}{q} + & (n-|C|) \cdot \log \dfrac{1}{1-q} \\
& = |C|\cdot \log \dfrac{1-q}{q} + n \cdot \log\dfrac{1}{1-q},
\end{align*}
for $0<q<1/2$. Here, $q$ can be loosely interpreted as the expected probability that a random node is part of $C$, according to the user. Typically, $q$ is to be chosen small.
To communicate the $|C|$ numbers $\ell_e$, in practice a fixed-length encoding (e.g., floating-point) would be used, and arguably it is also in this way (i.e., a fixed number of significant digits) that an analyst would assimilate the information.
This implies a further cost that increases linearly with $|C|$.
Hence, the DL of a cycle pattern, including a specification of $v \in C$ and the lower bounds $\ell_e$ for all $e\in C$, is of the form:
\begin{align*}
\text{DL}(C) = \alpha|C|+n\beta.
\end{align*}
where $\alpha > 0$ and $\beta > 0$ are defined as 
\begin{align}\label{eq:alphaBeta}
\alpha = \log \dfrac{1-q}{q} \text{,   } \beta = \log\dfrac{1}{1-q}.
\end{align}
We now formally define the subjective interestingess of a cycle pattern. 
\begin{definition}[Subjective Interestingness]
Given a directed graph $G = (V,E)$ with non-negative edge weights $w$, and parameters $\alpha > 0$ and $\beta > 0$, the subjective interestingness $F(C)$ of a cycle $C$ is defined\footnote{We note that $F(\emptyset) = 0$.} as: 

\begin{align}
\label{fobjexp}
F(C) = \dfrac{\text{IC(C)}}{\text{DL(C)}} = \dfrac{\sum_{e \in C} w(e)}{\alpha |C|+ n\beta}.
\end{align}
\end{definition}
%


\subsection{Modeling a user's prior beliefs}\label{sec:background}

As argued by \citet{debie2011}, a good choice for the background distribution $P$ is the maximum entropy distribution, subject to particular user expectations as linear constraints on the distribution.
Here, the domain of the distribution $P$ is the set of all possible edges over a given set of nodes.
For a better understanding of these models, we recap some existing results and discuss a toyexample below.

\subsubsection{A prior on the weighted in- and out-degrees}\label{sec:rowandcolumn}

In the case of a prior belief on the weighted in- and out-degree of each node, the distribution $P$ factorizes as a product of independent geometric distributions, one for each node pair.
As discussed in \citep{Deb:10a}, using a background distribution with the empirically weighted in- and out-degrees as constraints will ensure that cycle patterns are more interesting if they involve edges from low out-degree nodes to low in-degree nodes.
As it is quite common that weighted node degrees are well-understood (e.g., biologists have a good idea about the predatory component of the diet of different species in a food web), this is an important type of background distribution in practice.



\subsubsection{Additional priors on the density of any subsets}\label{sec:subsetdensity}

Additionally, extra constraints on the density on a number of user-provided subgraphs can be incorporated. For example, an economist might have knowledge of high trading volume between a group of neighboring countries, e.g., due to a free trade agreement or a common market, or a user might know that no self-edges exist in a network.
In this case, if an edge $e=(i,j)$ is part of the specified subgraph, the probability that this edge has a weight at least $\ell$ becomes:
\begin{equation*}
\label{eq:probgeosubg}
\text{Pr}(\mu(e) \geq \ell) = \exp\big(-\ell(\lambda^{\text{out}}_i+\lambda^{\text{in}}_j+\lambda^{\text{block}})\big),
\end{equation*} 
where $\lambda^{\text{out}}_i+\lambda^{\text{in}}_j+\lambda^{\text{block}}>0$. Here, $\lambda^{\text{out}}_i$ and $\lambda^{\text{in}}_j$ denote the Lagrange multipliers associated with the resp. row- and column sums of node $i$ and $j$, and $\lambda^{\text{block}}$ denotes the Lagrange multiplier associated with the density of the specified subgraph.
\citet{Adriaens2019} showed how these multipliers can still be computed efficiently for large networks and a limited number of specified subgraphs.
Figure~\ref{fig:maxent} shows an example of fitting the MaxEnt model on a 4x4 adjacency matrix $\mathbf{A}$ with different types of constraints.
It illustrates how adding more constraints results in a closer fit of the background distribution to the empirical network.
The probability $\text{Pr}(\mathbf{A}_{12}\geq 99) = 0.038$ for (a), 0.054 for (b) and 0.53 for (c). 

\begin{figure}
\centering
\begin{adjustbox}{minipage=\linewidth,scale={0.73}{0.72}}
\subfloat{
\begin{tikzpicture}[node distance =0.6cm, proc/.style={shape=ellipse, draw},baseline][trim axis left,trim axis right]\notag
\node [draw=none] (x11) {\small0};
\node [draw=none] (x12) [right of = x11] {\small99};
\node [draw=none] (x13) [right of = x12] {\small1};
\node [draw=none] (x14) [right of = x13] {\small0};

\node [draw=none] (x21) [below of = x11] {\small97};
\node [draw=none] (x22) [right of = x21] {\small0};
\node [draw=none] (x23) [right of = x22] {\small1};
\node [draw=none] (x24) [right of = x23] {\small2};

\node [draw=none] (x31) [below of = x21] {\small1};
\node [draw=none] (x32) [right of = x31] {\small1};
\node [draw=none] (x33) [right of = x32] {\small0};
\node [draw=none] (x34) [right of = x33] {\small98};

\node [draw=none] (x41) [below of = x31] {\small2};
\node [draw=none] (x42) [right of = x41] {\small0};
\node [draw=none] (x43) [right of = x42] {\small98};
\node [draw=none] (x44) [right of = x43] {\small0};

\draw [black] (-0.1,-2) to [square left brace] (-0.1,0.25);
\draw [black] (1.9,-2) to [square right brace] (1.9, 0.25);
\end{tikzpicture}
}\addtocounter{subfigure}{-1}
\subfloat{
\begin{tikzpicture}[node distance =0.6cm,proc/.style={shape=ellipse, draw}, baseline][trim axis left,trim axis right]
\node [draw=none] (x11) {\small0};
\node [draw=none] (x12) [right of = x11] {\small99};
\node [draw=none] (x13) [right of = x12] {\small1};
\node [draw=none] (x14) [right of = x13] {\small0};

\node [draw=none] (x21) [below of = x11] {\small97};
\node [draw=none] (x22) [right of = x21] {\small0};
\node [draw=none] (x23) [right of = x22] {\small1};
\node [draw=none] (x24) [right of = x23] {\small2};

\node [draw=none] (x31) [below of = x21] {\small1};
\node [draw=none] (x32) [right of = x31] {\small1};
\node [draw=none] (x33) [right of = x32] {\small0};
\node [draw=none] (x34) [right of = x33] {\small98};

\node [draw=none] (x41) [below of = x31] {\small2};
\node [draw=none] (x42) [right of = x41] {\small0};
\node [draw=none] (x43) [right of = x42] {\small98};
\node [draw=none] (x44) [right of = x43] {\small0};

\node [rotate around={-45:(0,0)},fit=(x11)(x44), proc, thick, black, xscale=0.95,yscale=0.1,label={[xshift=-1.2cm, yshift=1cm]$S_1$}] {};

\draw [black] (-0.1,-2) to [square left brace] (-0.1,0.25);
\draw [black] (1.9,-2) to [square right brace] (1.9, 0.25);
\end{tikzpicture}
}\addtocounter{subfigure}{-1}
\subfloat{
\begin{tikzpicture}[node distance =0.6cm,proc/.style={shape=ellipse, draw}, baseline][trim axis left,trim axis right]
\node [draw=none] (x11) {\small0};
\node [draw=none] (x12) [right of = x11] {\small99};
\node [draw=none] (x13) [right of = x12] {\small1};
\node [draw=none] (x14) [right of = x13] {\small0};

\node [draw=none] (x21) [below of = x11] {\small97};
\node [draw=none] (x22) [right of = x21] {\small0};
\node [draw=none] (x23) [right of = x22] {\small1};
\node [draw=none] (x24) [right of = x23] {\small2};

\node [draw=none] (x31) [below of = x21] {\small1};
\node [draw=none] (x32) [right of = x31] {\small1};
\node [draw=none] (x33) [right of = x32] {\small0};
\node [draw=none] (x34) [right of = x33] {\small98};

\node [draw=none] (x41) [below of = x31] {\small2};
\node [draw=none] (x42) [right of = x41] {\small0};
\node [draw=none] (x43) [right of = x42] {\small98};
\node [draw=none] (x44) [right of = x43] {\small0};

\node [rotate around={-45:(0,0)},fit=(x11)(x44), proc, thick, black, xscale=0.95,yscale=0.1,label={[xshift=-1.2cm, yshift=1cm]$S_1$}] {};
\draw (0.3cm,-0.3cm) node[rectangle, thick, minimum height=1cm,minimum width=1cm,draw,label={[xshift=0.5cm, yshift=0cm]$S_2$}] {};

\draw [black] (-0.1,-2) to [square left brace] (-0.1,0.25);
\draw [black] (1.9,-2) to [square right brace] (1.9, 0.25);
\end{tikzpicture}
}\addtocounter{subfigure}{-1}

\subfloat[]{
\begin{tikzpicture}[node distance =0.6cm][trim axis left,trim axis right]
\node [draw=none] (x11) {\small25};
\node [draw=none] (x12) [right of = x11] {\small25};
\node [draw=none] (x13) [right of = x12] {\small25};
\node [draw=none] (x14) [right of = x13] {\small25};

\node [draw=none] (x21) [below of = x11] {\small25};
\node [draw=none] (x22) [right of = x21] {\small25};
\node [draw=none] (x23) [right of = x22] {\small25};
\node [draw=none] (x24) [right of = x23] {\small25};

\node [draw=none] (x31) [below of = x21] {\small25};
\node [draw=none] (x32) [right of = x31] {\small25};
\node [draw=none] (x33) [right of = x32] {\small25};
\node [draw=none] (x34) [right of = x33] {\small25};

\node [draw=none] (x41) [below of = x31] {\small25};
\node [draw=none] (x42) [right of = x41] {\small25};
\node [draw=none] (x43) [right of = x42] {\small25};
\node [draw=none] (x44) [right of = x43] {\small25};

\draw [black] (-0.1,-2) to [square left brace] (-0.1,0.25);
\draw [black] (1.9,-2) to [square right brace] (1.9, 0.25);
\end{tikzpicture}
}
\subfloat[]{
\begin{tikzpicture}[node distance =0.6cm][trim axis left,trim axis right]
\node [draw=none] (x11) {\small0};
\node [draw=none] (x12) [right of = x11] {\small33.3};
\node [draw=none] (x13) [right of = x12] {\small33.3};
\node [draw=none] (x14) [right of = x13] {\small33.3};

\node [draw=none] (x21) [below of = x11] {\small33.3};
\node [draw=none] (x22) [right of = x21] {\small0};
\node [draw=none] (x23) [right of = x22] {\small33.3};
\node [draw=none] (x24) [right of = x23] {\small33.3};

\node [draw=none] (x31) [below of = x21] {\small33.3};
\node [draw=none] (x32) [right of = x31] {\small33.3};
\node [draw=none] (x33) [right of = x32] {\small0};
\node [draw=none] (x34) [right of = x33] {\small33.3};

\node [draw=none] (x41) [below of = x31] {\small33.3};
\node [draw=none] (x42) [right of = x41] {\small33.3};
\node [draw=none] (x43) [right of = x42] {\small33.3};
\node [draw=none] (x44) [right of = x43] {\small0};

\draw [black] (-0.1,-2) to [square left brace] (-0.1,0.25);
\draw [black] (1.9,-2) to [square right brace] (1.9, 0.25);
\end{tikzpicture}
}
\subfloat[]{
\begin{tikzpicture}[node distance =0.6cm][trim axis left,trim axis right]
\node [draw=none] (x11) {\small0};
\node [draw=none] (x12) [right of = x11] {\small98};
\node [draw=none] (x13) [right of = x12] {\small1};
\node [draw=none] (x14) [right of = x13] {\small1};

\node [draw=none] (x21) [below of = x11] {\small98};
\node [draw=none] (x22) [right of = x21] {\small0};
\node [draw=none] (x23) [right of = x22] {\small1};
\node [draw=none] (x24) [right of = x23] {\small1};

\node [draw=none] (x31) [below of = x21] {\small1};
\node [draw=none] (x32) [right of = x31] {\small1};
\node [draw=none] (x33) [right of = x32] {\small0};
\node [draw=none] (x34) [right of = x33] {\small98};

\node [draw=none] (x41) [below of = x31] {\small1};
\node [draw=none] (x42) [right of = x41] {\small1};
\node [draw=none] (x43) [right of = x42] {\small98};
\node [draw=none] (x44) [right of = x43] {\small0};

\draw [black] (-0.1,-2) to [square left brace] (-0.1,0.25);
\draw [black] (1.9,-2) to [square right brace] (1.9, 0.25);
\end{tikzpicture}
}
\end{adjustbox}
\caption{(top row) A toy graph, with constraints on the in- and out degrees of each node (a, b and c), combined with constraints on the densities of the sets $S_1$ (b and c) and $S_2$ (c). (bottom row) The expected values of the edges according to the MaxEnt distribution. \label{fig:maxent}}
\end{figure}
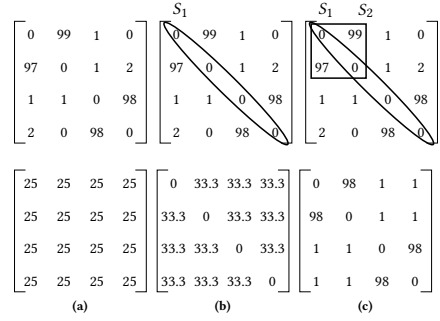

\section{Problem definition}
\label{sec:problem}
%
%


The first problem considered in this paper is the problem of finding the ``Maximum Subjectively Interesting Cycle'' in a graph.\footnote{Note that although the problem appears to have two parameters, in reality this can be reduced to one, e.g. by multiplying the objective with $\beta$ and substituting $\alpha/\beta$ with a single parameter $\gamma$.} Formally:
\begin{problem}[\msic]
\label{fobj}
Given a directed graph $G$ with non-negative edge weights, and parameters $\alpha,\beta>0$, find a simple cycle $C$ such that $F(C)$ is maximized.
\end{problem}

Additionally, we can constrain the cycle to include a given set of terminal nodes to find ``Maximum Subjectively Interesting Steiner Cycle''. This leads to the second problem we address in this paper:
\vspace{-5pt}
\begin{problem}[\kmsic]
\label{fobjquery}
Given a directed graph $G$ with non-negative edge weights, a set of $k$ terminal nodes, and parameters $\alpha,\beta>0$, find a simple cycle $C$ such that $C$ contains all the terminals and $F(C)$ is maximized.
\end{problem}

\msic is closely related to two well-known graph problems. For $\alpha = 0$, \msic is equivalent to the problem of finding the \emph{longest cycle in a digraph}, an \NP-hard problem that is known for its difficulty to approximate \citep{bjorklund,Gabowcycles}. On the other hand, for $\beta = 0$, \msic is equivalent to the problem of finding a \emph{maximum mean-cycle} in a directed graph with non-negative edge weights. This problem can be solved in polynomial time by reversing the sign of the edge weights of the input graph and running Karp's minimum mean-cycle algorithm \citep{KARP}, that is originally devised to find minimum mean-cycle in digraphs with real-valued weights. Although \msic is closely related to a tractable and to an \NP-hard problem, it is not equivalent to either one as our problem setting assumes $\alpha > 0$ and $\beta > 0$. Yet, in Section~\ref{sec:theory} we show that \msic is \NP-hard (as the longest cycle problem), while we discuss how Karp's algorithm can be used to provide approximations. This is a plausible approach, as in practice $\alpha \gg \beta$ (it takes more effort to assimilate the fact that a node is part of a cycle pattern than that a node is not part of a cycle pattern), such that the interestingness measure is closer to the maximum mean-cycle objective than to the longest cycle objective.

Likewise, \kmsic is closely related to two Steiner cycle problem variants. For $\alpha = 0$, \kmsic is equivalent to the problem of finding a \emph{maximum Steiner cycle}, i.e. Steiner cycle with max. total weight, and for $\beta = 0$, \kmsic is equivalent to the problem of finding a \emph{maximum mean Steiner cycle} (\mmscp) in a digraph with non-negative edge weights. To the best of our knowledge, there are no known results on the approximability of either problems. Besides being \NP-hard, we show in the next section that neither of these Steiner cycle problems, nor \kmsic, can  be approximated within a ratio that is polynomial in the number of nodes.

%
%


\section{Computational Complexity}
\label{sec:theory}
Unsurprisingly, both \msic and \kmsic are \NP-hard problems. We show this in the following section.
Section~\ref{sec:inapprox} is dedicated to proving some strong inapproxability results.
We note that the reduction in Lemma~\ref{lem:approxquery2} can directly be applied to the Max. Steiner Cycle and Max. Mean Steiner Cycle problems, which is a novel result in itself.

\subsection{Hardness results}\label{sec:hard}
The hardness of both problems directly follows from Lemma~\ref{lem:approx1} and Lemma~\ref{lem:approxquery2}, but are presented here separately for completeness.
\begin{lemma}
\label{lem:hardness}
\msic is \NP-hard.
\end{lemma}
\begin{proof}
We use a reduction from the \NPcomplete Hamiltonian cycle problem. Given a digraph $G=(V,E)$, the Hamiltonian cycle problem is to determine whether $G$ has a simple cycle that visits every node  exactly once. Given an instance of the Hamiltonian cycle problem, we construct an instance of \msic by assigning a constant weight to every edge, $w(e) = \rho$, $\forall e \in E$. Then, for any cycle $C$, we have $F(C) = \frac{\rho |C|}{\alpha |C|+ n \beta}$. Notice that, in all the instances  of \msic with uniform edge weights, $F(C)$ monotonically increases with $|C|$, ceteris paribus, and obtains the maximum possible value whenever $|C| = n$. Thus, $F(C) = \frac{\rho n}{\alpha n+ n\beta}$ iff $C$ is a Hamiltonian cycle and $F(C)  = \frac{\rho |C|}{\alpha |C| + n\beta} < \frac{\rho n}{\alpha n+ n\beta}$ otherwise. Hence, we can use the solution to \msic to decide the solution to the Hamiltonian cycle problem. 
\end{proof}

\begin{lemma}
\kmsic is \NP-hard.
\end{lemma}
\begin{proof} 
We use a reduction from the \NPcomplete Hamiltonian cycle problem. Let $G=(V,E)$ be a given instance of the Hamiltonian cycle problem. We construct an instance of \kmsic by assigning $w(e) = 1, \forall e \in E$, and picking an arbitrary subset of $k$ nodes as the terminals. Given uniform weights, for any Steiner cycle $C$, $F(C)$ monotonically increases with $|C|$, ceteris paribus. Let $C^*$ denote the optimal solution to \kmsic.  Then, $G$ is a YES instance of the the Hamiltonian cycle problem iff $F(C^*) = \dfrac{n}{\alpha n + n \beta}$ and NO instance iff $F(C^*) < \dfrac{n}{\alpha n + n \beta}$.  
\end{proof}

\subsection{Inapproximability results}\label{sec:inapprox}

\begin{lemma}
\label{lem:approx1}
There exists no constant-factor polynomial-time approximation algorithm for \msic, unless $\Pclass = \NP$. 
\end{lemma}

\begin{proof}
To prove this, we use an approximation preserving reduction from the Longest Cycle problem in digraphs~\cite{bjorklund}. Specifically, we use an \emph{A-reduction} \citep{Crescenzi1997ASG} that preserves membership in \APX, which is the class of \NP optimization problems that admit polynomial-time constant-factor approximation algorithms. 

To show that a reduction from Longest Cycle problem to \msic is an \emph{A-reduction}, we need to show that (i) there exists a polynomial-time computable function $g$ mapping the solutions of \msic to the solutions of the Longest Cycle problem, and (ii) a polynomial-time computable function $c: \mathbb{Q} \cap (1, \infty) \rightarrow \mathbb{Q} \cap (1, \infty)$ such that any algorithm providing $r$-approximation to \msic with the approximate solution $C$ provides a $c(r)$-approximation to the Longest Cycle problem using the approximate solution $g(C)$. 

Let $G=(V,E)$ be a given an instance of the Longest Cycle problem in digraphs. We construct an instance of \msic by assigning a constant weight $w(e)=\rho$, $\forall e \in E$ in $G$. Assume there exists a polynomial-time algorithm $A$ which provides a $r$-approximation to \msic for some constant $r \ge 1$. Let $C^*$ denote the optimal solution to \msic and let $C_A$ denote the solution returned by algorithm $A$. Then we have, 
\begin{align}\label{eq:aReduction}
\dfrac{F(C^*)}{F(C_A)} = \dfrac{\rho |C^*|}{\rho |C_A|} \cdot \dfrac{\alpha |C_A| +n \beta}{\alpha |C^*| + n\beta } \le r
\end{align}
Reminding that $F(C)$ monotonically increases with $|C|$ in such instances of \msic with uniform edge weights, we define $g$ as the identity function, and use the solutions of \msic as the solutions of the Longest Cycle problem. Then, by re-arranging (\ref{eq:aReduction}) and using the fact that $2 \le |C| \le n$ for any cycle $C$, we have: 
\begin{align*}
\dfrac{|C^*|}{|C_A|} &\le r \cdot \dfrac{\alpha |C^*| +n \beta}{\alpha |C_A| + n\beta } \\
&\le r \cdot \dfrac{n (\alpha + \beta)}{2 \alpha + n\beta } \\
&\le r \cdot (1 + \alpha / \beta )
\end{align*}
We have just showed that the Longest Cycle problem is A-reducible to \msic. Finally, $\Pclass = \NP$, given that the Longest Cycle problem in digraphs is not in \APX \citep{bjorklund,Gabowcycles}, we conclude that \msic is also not in \APX. 
\end{proof}

\citet{bjorklund} show that there exists no polynomial-time approximation algorithm for the Longest Cycle problem in unweighted Hamiltonian digraphs with performance ratio $n^{1-\epsilon}$ for any fixed $\epsilon > 0$, unless $\Pclass = \NP$. Next we show the implications of this strong inapproximability result for solving \msic in Hamiltonian digraphs with uniform edge weights. 

\begin{lemma}
\label{lem:approx2}
It is \NP-hard to approximate \msic in a Hamiltonian digraph with uniform weights within a factor of 
$$\dfrac{n^{1-\epsilon} + \alpha / \beta}{1+ \alpha / \beta},$$ 
for any $\epsilon > 0$, unless $\Pclass = \NP$. 
\end{lemma}

\begin{proof}
Let $G=(V,E)$ be an unweighted Hamiltonian digraph denoting an instance of the Longest Cycle problem\cite{bjorklund}. Given $G = (V,E)$, we construct an instance of  \msic by assigning a constant weight to every edge, $w(e) = \rho$, $\forall e \in E$. Assume by contradiction that there exists such an approximation algorithm $A$ which finds a solution $C_A$ satisfying
\begin{align}\label{eq:inapprox2}
\dfrac{\rho |C_A|}{\alpha |C_A| + n \beta} \ge \dfrac{1+ \alpha / \beta}{n^{1-\epsilon} + \alpha / \beta} \cdot  \dfrac{\rho n}{\alpha n + n \beta}
\end{align}
By re-arranging the terms in (\ref{eq:inapprox2}), we obtain $|C_A| \ge n^{\epsilon}$ 
implying that any such approximation algorithm to \msic leads to a polynomial-time $n^{1-\epsilon}$-approximation algorithm for the Longest Cycle problem in unweighted Hamiltonian digraphs, which is a contradiction, unless $\Pclass = \NP$.  
\end{proof}




Next we show the hardness of approximating \kmsic. 

\begin{lemma}
\label{lem:approxquery2}
It is \NP-hard to approximate \kmsic within a factor polynomial in the input size in digraphs with non-negative edge weights for any $k \ge 1$, unless $\Pclass = \NP$. 
\end{lemma}

\begin{proof}
To prove this, we use a reduction from the \NP-complete Restricted Two node Disjoint Paths problem (R2VDP), which was introduced by \citet{bjorklund} as the restricted version of the Two node Disjoint Paths problem (2VDP)~\cite{Perl1978FindingTD}. Given a digraph of order $n \ge 4$ and four nodes, 2VDP problem seeks to determine whether there exist two node disjoint paths, one from node $1$ to $2$ and one from node $3$ to $4$. In the restricted version R2VDP of 2VDP, all the YES instances of 2VDP are guaranteed to contain two such paths that together exhaust all nodes of $G$, i.e., the graph $G$ with the additional edges from node $2$ to $3$ and from $4$ to $1$, contains a Hamiltonian cycle through these edges in YES instances to R2VDP. 

Assume that there exists an approximation algorithm for \kmsic with ratio $p(n) \ge 1$ that is a polynomial of $n$. We show how to decide R2VDP by using such algorithm with approximation ratio $p(n)$. Given an instance of R2VDP, we construct an instance of \kmsic as follows. We connect 2 copies $G_1$ and $G_2$ of $G$ by adding edges (i) from node $2$ in $G_1$ to node $1$ in $G_2$, and (ii) from node $4$ in $G_2$ to node $3$ in $G_1$. We also add an edge $(4,1)$ in $G_1$ and an edge $(2,3)$ in $G_2$. For each edge we assign a weight of $1$, except for the edge $(2,3)$ in $G_2$ for which we assign a weight of $W = n \cdot p(n) + 1$. Finally, we set the node $1$ of $G_1$ as the terminal for $1$-\msic. Let $G' = (V', E')$ denote the resulting graph, as shown in Figure~\ref{fig:approxred}.

Let $C^*$ denote the optimal solution to $1$-\msic in $G'$. If $G$ is a YES instance of R2VDP, then $C^*$ is a Hamiltonian cycle in $G'$, containing $2n$ edges with a total weight of $2n  + n \cdot p(n)$,  since,  $F(C^*) = \dfrac{2n  + n \cdot p(n)}{2 \alpha n + n \beta} > \dfrac{|C|}{\alpha |C| + n \beta} $ for any other Steiner cycle $C$ that is not Hamiltonian, thus, not containing the edge $(2,3)$ in $G_2$. On the other hand, if $G$ is a NO instance to R2VDP, then $C^*$ can have at most $2n - 2$ edges, excluding the edge $(4,1)$ in $G_1$ and the edge $(2,3)$ in $G_2$, thus, $\dfrac{|C^*|}{|C^*| + 1} \le \dfrac{2n - 2}{\alpha (2n-2) + n \beta}$. 

We have just shown that, unless $\Pclass = \NP$, it is not possible to approximate $1$-\msic within a factor that is polynomial in the input size in digraphs with non-negative edge weights. It is easy to see that as $k$ increases, the problem only becomes harder, with $k=n$ corresponding to the search for a Hamiltonian cycle. Thus, the result follows for any $k \ge 1$. 
\end{proof}

\begin{corollary}
It is \NP-hard to approximate Maximum Steiner Cycle and Maximum-Mean Steiner Cycle problems within a factor polynomial in the input size in digraphs with non-negative edge weights for any $k \ge 1$, unless $\Pclass = \NP$. 
\end{corollary}

\begin{proof}
The results directly follows from the reduction given in Lemma~\ref{lem:approxquery2}. 
\end{proof}

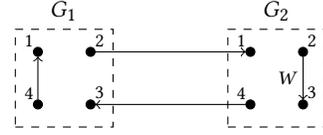
\begin{figure}
\centering
\begin{tikzpicture}[node distance =1.75cm]
\node [draw, fill, circle, scale = 0.4, label={[xshift=-0.12cm, yshift=-0.1cm]\footnotesize$1$}] (1) {};
\node [draw, fill, circle, scale = 0.4, label={[xshift=0.12cm, yshift=-0.1cm]\footnotesize$2$}] (2) [right of  = 1] {};
\node [draw, fill, circle, scale = 0.4, label={[xshift=-0.12cm, yshift=-0.1cm]\footnotesize$4$}] (4) [below of = 1] {};
\node [draw, fill, circle, scale = 0.4, label={[xshift=0.12cm, yshift=-0.1cm]\footnotesize$3$}] (3) [below of = 2] {};

\node [draw, fill, circle, scale = 0.4, label={[xshift=-0.12cm, yshift=-0.1cm]\footnotesize$1$}] (1a) [right = 2cm of 2]{};
\node [draw, fill, circle, scale = 0.4, label={[xshift=0.12cm, yshift=-0.1cm]\footnotesize$2$}] (2a) [right of  = 1a] {};
\node [draw, fill, circle, scale = 0.4, label={[xshift=-0.12cm, yshift=-0.1cm]\footnotesize$4$}] (4a) [below of = 1a] {};
\node [draw, fill, circle, scale = 0.4, label={[xshift=0.12cm, yshift=-0.1cm]\footnotesize$3$}] (3a) [below of = 2a] {};

\draw (0.35cm,-0.35cm) node[rectangle,dashed, minimum height=1.3cm,minimum width=1.3cm,draw,label={[xshift=0cm, yshift=0cm]$G_1$}] {};
\draw (3.175cm,-0.35cm) node[rectangle,dashed, minimum height=1.3cm,minimum width=1.3cm,draw,label={[xshift=0cm, yshift=0cm]$G_2$}] {};

\draw [->] (2) to (1a);
\draw [->] (2a) -- node[midway, above left = -0.2cm and -0.05 cm] {\footnotesize$W$}  (3a);
\draw [->] (4)  to (1);
\draw [->] (4a) to (3);

\end{tikzpicture}
\caption{A visualization of the construction in Lemma~\ref{lem:approxquery2}.\label{fig:approxred}}
\end{figure}

\section{Practical Algorithms}
\label{sec:algorithms}
\subsection{Algorithms for MSIC}\label{sec:methmsic}
\subsubsection{Karp's Algorithm}
Due to the \NP-hardness of our problem, we resort to the \emph{maximum mean cycle} as an approximate solution to \msic. We first note that the maximum mean cycle in a graph $G$ is equivalent to the minimum mean cycle in the graph $G'$ obtained by reversing the sign of the edge weights of $G$. 

The problem of finding the \emph{minimum mean cycle} (\mmc) in a graph with real-valued edge weights is well-studied in the literature and admits efficient polynomial algorithms as shown by Karp~\citep{KARP}. Karp's \mmc algorithm runs in $\Theta(nm)$ time on \emph{any} instance. As noted by \citet{Dasdan}, there are other algorithms, with worse theoretical bounds, performing significantly better in practice, such as, Howard's algorithm \citep{howard:dp} and Young's algorithm \citep{Young1991FasterPS}. \citet{Dasdan} have given excellent survey of the different algorithms and their performance in practice. 

In this paper we use Karp's \mmc algorithm, not only due to its ease of implementation but also it still holds one of the best asymptotic running times. 

We briefly review Karp's algorithm for completeness. First, recall that in the \mmc problem we are given a directed graph $G=(V,E)$
and an arbitrary edge-weight function  $w: E \to \mathbb{R}$, and the goal is to find a cycle $C$ in $G$ that minimizes the average weight, i.e., 
$\rho^* = \min_{\{C \text{ cycle of } G\}}\left\{  \frac{\sum_{e\in C}w(e)}{|C|} \right\}$.
Karp provided the following elegant characterization for the weight of the minimum mean-cycle:
\begin{equation}
\label{equation:Karp}
\rho^* = \min_{v\in V} \max_{1\le k \le n}
\frac{D_n(v) - D_k(v)}{n-k},
\end{equation}
where $D_k(v)$ is the minimum weight of an edge progression\footnote{Both edges and nodes may be repeated.} of length $k$ from an arbitrary source $s$,
while $D_k(v)=+\infty$ if no such path exists. It is assumed that all nodes are reachable from $s$.
The algorithm computes $\rho^*$ using Equation~(\ref{equation:Karp})
after computing the values $D_k(v)$ for all $v\in V$ and $k=1,\ldots,n$ 
via the recurrence
$D_k(v)=\min_{(u,v)\in E}\{D_{k-1}(u) + w(u,v) \}$, 
initialized with $D_0(s)=0$ and $D_0(v)=+\infty$ for $v\not= s$.
The actual cycle can be extracted by traversing the edge progression $D_n(v)$, for the node $v$ that minimizes (\ref{equation:Karp}).

Obviously, the \mmc and \msic problems optimize different objectives, 
however, we can use Karp's algorithm as a heuristic for the \msic problem.
We can show that the cycle with maximum mean weight provides a $\bigO(n)$-approximation 
for \msic in arbitrary graphs with non-negative edge weights.

\begin{lemma}
\label{karpapp}
Karp's \mmc algorithm~\cite{KARP} provides a $\bigO(n)$-approximation for \msic in arbitrary graphs with non-negative edge weights. 
\end{lemma}

\begin{proof}
Given a directed graph $G=(V,E)$ with non-negative weights, let $C_K$ denote the cycle with maximum mean weight and let $C^*$ denote the optimal solution to \msic. Then, by using the fact that 
$\sum_{e \in C_K} w(e) / {|C_K|} \ge \sum_{e \in C^*} w(e)/{|C^*|}$, and that $2 \le |C| \le n$ for any cycle $C$, we obtain
\begin{align}
\label{eq:boundkarp}
&\dfrac{F(C^*)}{F(C_K)} = \dfrac{\sum_{e \in C^*} w(e)}{\alpha |C^*| + n\beta } \cdot \dfrac{\alpha |C_K| + n\beta }{\sum_{e \in C_K} w(e)} \notag\\ 
&\le \dfrac{\sum_{e \in C_K} w(e) \cdot \frac{|C^*|}{|C_K|}}{\alpha |C^*| + n\beta } \cdot \dfrac{\alpha |C_K| + n\beta }{\sum_{e \in C_K} w(e)} \notag\\ 
&= \dfrac{\alpha + n\beta  / |C_K|}{\alpha + n\beta  / |C^*|} \le \dfrac{\alpha + n\beta  / 2}{\alpha + \beta}
\end{align}
\end{proof}

In Section~\ref{sec:experiments} we show that the bound (\ref{eq:boundkarp}) is quite good in practice, as long as the parameter $q$, i.e., the expected probability that a random node is part of a cycle, is small. Roughly put, small indicates not more than the density of the network. Notice that, as $q$ increases, the value of $\beta$ also increases, see (\ref{eq:alphaBeta}), and thus it becomes more likely that the optimal solution to \msic is the longest cycle in the graph, which is hard to approximate. 

\subsubsection{A variant of Karp's algorithm}\label{sec:algo:var}
Although efficient, a direct application of Karp's algorithm to solve \msic disregards the information about the parameters $\alpha$ and $\beta$. Thus, we propose a natural extension of Karp's algoritm that incorporates the role of the parameters $\alpha$ and $\beta$ aligned with the objective function of \msic. To this end, we modify Karp's algorithm to find the node $v$ that minimizes (on the edge-signs reversed graph $G'$) the following: 
\begin{equation}
\label{equation:Karpvar}
\min_{v\in V} \max_{1\le k \le n}
\frac{D_n(v) - D_k(v)}{\alpha(n-k)+n\beta}.
\end{equation}
Notice that, as in Karp's characterization, the numerator in (\ref{equation:Karpvar}) mimics the weight of a cycle of length $(n-k)$ found for each $v \in V$, so (\ref{equation:Karpvar}) operates with the objective function of \msic. Similar to Karp's algorithm, this algorithm runs in $\Theta(nm)$ time and the cycle for the minimizer $v$ can be found by traversing the edge progression $D_n(v)$. 


\subsection{Algorithms for $k$-MSIC}\label{sec:methkmsic}
The \kmsic problem is reminiscent of Steiner cycle problems, thus, one could consider the solutions of related problems, such as  maximum mean Steiner cycle (\mmscp), for approximating \kmsic. However, as we have shown in Section~\ref{sec:theory}, besides being \NP-hard, both problems cannot be approximated within a ratio that is polynomial in the number of nodes. 
Existing algorithms for approximating Steiner cycle problem variants are less well-known, and in most cases these algorithms have strict requirements as we review next.

\citet{Ste10} proposed a $\frac{3}{2}\log_2(k)$-approximation algorithm for the {\em mininum} Steiner cycle problem 
on $k$ terminal nodes in non-negatively weighted graphs in which the edge weights satisfy the triangle inequality.
We note that not only the instances of \kmsic do not necessarily satisfy the triangle inequality but also we study a maximization problem.


\citet{SALAZARSteinercycle} introduced a {\em minimum} Steiner cycle problem variant with node penalties and considered a 0-1 integer linear program examining the Steiner cycle polytope. Besides having a different context, their method is of theoretical interest that doesn't translate into practical algorithms for \kmsic. 

\citet{Kanellakis} propose a \emph{local search} method for directed TSP, extending the Lin-Kernighan heuristic proposed for undirected TSP~\citep{Lin}. We adopt the local search approach proposed by \citet{Kanellakis} for directed TSP and extend their techniques for finding Steiner cycles of interest. We will refer to our local search heuristic for \kmsic as \localSC. 


The local search method by \citet{Kanellakis} starts with a random initial solution then considers the so-called ``sequential primary'' and ``quad'' changes. 
In a sequential primary change, three edges $(a,b)$, $(c,d)$, and $(e,f)$, encountered in this order on the cycle, 
are removed from the cycle, and the edges $(a,d)$, $(c,f)$ and $(e,b)$ are added. 
In a quad change, the rewiring consists of removing four edges and reconnecting opposite edges, 
as shown in Figure~\ref{fig:changes}(b). The neighborhood of each step in their local search consists of a cost-dependent subset, determined by a number of heuristic rules. 
The search stops when no significant improvements can be made. 

When transforming this search from a TSP setting to a Steiner cycle setting, a few adjustments have to be made. 
Besides the primary and quad change, we propose two new changes in \localSC. 
The {\em shortcutting change} shortcuts the initial solution into a smaller Steiner cycle. 
The {\em extending change} bypasses an edge in a Steiner cycle, by replacing the edge with two new edges. A visualization of all the changes considered by \localSC are provided in Figure~\ref{fig:changes}. 

Given a set $Q$ of $k$ terminal nodes and an upper bound $\lenmax \ge k$ on the cycle length, \localSC finds an initial Steiner cycle of $G$ as follows: 
\begin{enumerate}
\item Prune $G$ by only considering nodes $v \in V$ s.t. 
\begin{equation*}
\forall q \in Q: \, \ell(q \rightsquigarrow v)+\ell(v \rightsquigarrow q) \leq \lenmax,
\end{equation*}
where $\ell(\cdot)$ denotes the (unweighted) shortest path length. This step can be performed in time $\bigO(k(n+m))$.
\item Run a randomized depth-first search to find an initial valid Steiner cycle. The search is guided by a heuristic, and each $v$ that has a low total distance 
towards all query nodes has a higher chance of being explored first, i.e., at any time in the depth-first search, 
the probability that a node $v$ is chosen from the stack is proportional to 
$1/\sum_{q \in Q} \ell(v \rightsquigarrow q)$.
\end{enumerate}
After \localSC finds an initial Steiner cycle, a sequence of changes depicted Figure~\ref{fig:changes} are applied. When considering a type of change, \localSC always selects the one that yields the largest improvement to the objective function (\ref{fobjexp}). \localSC first applies a number of extending changes to the initially found cycle until the cycle length is equal to $\lenmax$. 
Then, \localSC greedily keeps selecting the best change among the sequantial, quad, or shortcutting changes until no improvements can be made.
If \localSC doesn't return a solution, then no Steiner of length at most $\lenmax$ exists for the given $k$ terminal nodes. The idea is to run this randomized procedure a couple of times (1-5 in the experiments), and pick the best solution. 

Unlike the method of \citet{Kanellakis}, a neighborhood in our local search will consist of \emph{all} the possible changes. For a Steiner cycle of length $\lenmax$, there are $\bigO({\lenmax^2})$ shortcutting changes, $\bigO(n \cdot \lenmax)$ extending changes, $\bigO({\lenmax^3})$ primary changes and $\bigO({\lenmax^4})$ quad changes, which is feasible to evaluate for a reasonable upper bound $\lenmax$.

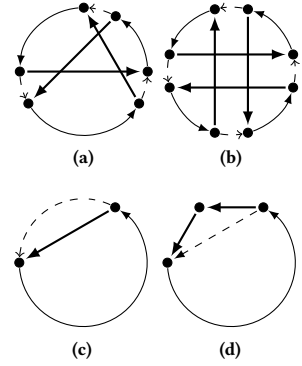
\begin{figure}
\captionsetup[subfigure]{}
\centering

\subfloat[]{
\begin{tikzpicture}[node distance =1.75cm]
\def \margin {5}
\def \radius {0.85cm}

\node[draw, fill, circle, scale = 0.4] at (90:\radius) {};
\node[draw, fill, circle, scale = 0.4] at (60:\radius) {};
\node[draw, fill, circle, scale = 0.4] at (0:\radius) {};
\node[draw, fill, circle, scale = 0.4] at (-30:\radius) {};
\node[draw, fill, circle, scale = 0.4] at (180:\radius) {};
\node[draw, fill, circle, scale = 0.4] at (210:\radius) {};

\draw[->, >=latex] (\margin:\radius) 
    arc (\margin:60-\margin:\radius);
\draw[->, >=latex] (90+\margin:\radius) 
    arc (90+\margin:180-\margin:\radius);
\draw[->, >=latex] (210+\margin:\radius) 
    arc (210+\margin:330-\margin:\radius);

\draw[->, dashed] (60+\margin:\radius) 
    arc (60+\margin:90-\margin:\radius);
\draw[->, dashed] (180+\margin:\radius) 
    arc (180+\margin:210-\margin:\radius);
\draw[->,, dashed] (330+\margin:\radius) 
    arc (330+\margin:360-\margin:\radius);

\draw[shorten >=0.1cm ,shorten <=0.1cm,->,>=latex,thick] (180:\radius) -- (0:\radius);
\draw[shorten >=0.1cm ,shorten <=0.1cm,->,>=latex,thick] (330:\radius) -- (90:\radius);
\draw[shorten >=0.1cm ,shorten <=0.1cm,->,>=latex,thick] (60:\radius) -- (210:\radius);
\end{tikzpicture}
}
\subfloat[]{
\begin{tikzpicture}[node distance =1.75cm]
\def \margin {5}
\def \radius {0.85cm}

\node[draw, fill, circle, scale = 0.4] at (-15:\radius) {};
\node[draw, fill, circle, scale = 0.4] at (15:\radius) {};
\node[draw, fill, circle, scale = 0.4] at (75:\radius) {};
\node[draw, fill, circle, scale = 0.4] at (105:\radius) {};
\node[draw, fill, circle, scale = 0.4] at (165:\radius) {};
\node[draw, fill, circle, scale = 0.4] at (195:\radius) {};
\node[draw, fill, circle, scale = 0.4] at (255:\radius) {};
\node[draw, fill, circle, scale = 0.4] at (285:\radius) {};

\draw[shorten >=0.1cm ,shorten <=0.1cm,->,>=latex,thick] (165:\radius) -- (15:\radius);
\draw[shorten >=0.1cm ,shorten <=0.1cm,->,>=latex,thick] (-15:\radius) -- (195:\radius);

\draw[shorten >=0.1cm ,shorten <=0.1cm,->,>=latex,thick] (255:\radius) -- (105:\radius);
\draw[shorten >=0.1cm ,shorten <=0.1cm,->,>=latex,thick] (75:\radius) -- (285:\radius);

\draw[->, >=latex] (15+\margin:\radius) 
    arc (15+\margin:75-\margin:\radius);
\draw[->, >=latex] (105+\margin:\radius) 
    arc (105+\margin:165-\margin:\radius);
\draw[->, >=latex] (195+\margin:\radius) 
    arc (195+\margin:255-\margin:\radius);
\draw[->, >=latex] (285+\margin:\radius) 
    arc (285+\margin:345-\margin:\radius);

\draw[->, dashed] (345+\margin:\radius) 
    arc (345+\margin:375-\margin:\radius);
\draw[->, dashed] (75+\margin:\radius) 
    arc (75+\margin:105-\margin:\radius);
\draw[->, dashed] (165+\margin:\radius) 
    arc (165+\margin:195-\margin:\radius);
\draw[->, dashed] (255+\margin:\radius) 
    arc (255+\margin:285-\margin:\radius);

\end{tikzpicture}
}

\subfloat[]{
\begin{tikzpicture}[node distance =1.75cm]
\def \margin {5}
\def \radius {0.85cm}

\node[draw, fill, circle, scale = 0.4] at (60:\radius) {};
\node[draw, fill, circle, scale = 0.4] at (180:\radius) {};

\draw[->, >=latex] (180+\margin:\radius) 
    arc (180+\margin:420-\margin:\radius);
\draw[->, dashed] (90+\margin:\radius) 
    arc (90+\margin:180-\margin:\radius);

\draw[-, dashed] (60+\margin:\radius) 
    arc (60+\margin:90-\margin:\radius);

\draw[shorten >=0.1cm ,shorten <=0.1cm,->,>=latex,thick] (60:\radius) -- (180:\radius);

\end{tikzpicture}
}
\subfloat[]{
\begin{tikzpicture}[node distance =1.75cm]
\def \margin {5}
\def \radius {0.85cm}

\node[draw, fill, circle, scale = 0.4] at (60:\radius) {};
\node[draw, fill, circle, scale = 0.4] at (180:\radius) {};
\node[draw, fill, circle, scale = 0.4] at (120:\radius) {};

\draw[->, >=latex] (180+\margin:\radius) 
    arc (180+\margin:420-\margin:\radius);

\draw[shorten >=0.1cm ,shorten <=0.1cm,->,>=latex,dashed] (60:\radius) -- (180:\radius);
\draw[shorten >=0.1cm ,shorten <=0.1cm,->,>=latex,thick] (60:\radius) -- (120:\radius);
\draw[shorten >=0.1cm ,shorten <=0.1cm,->,>=latex,thick] (120:\radius) -- (180:\radius);

\end{tikzpicture}
}
\caption{(a) Sequential primary change (b) Quad change (c) Shortcutting change (d) Extending change.\label{fig:changes}}
\end{figure}

\section{Experiments}
\label{sec:experiments}
The goal of this section is manifold. First, we would like to evaluate the quality of solutions obtained by Karp's \mmc algorithm, the variant from Section~\ref{sec:algo:var} and our local Steiner cycle search heuristic \localSC.
To this end, we conduct experiments on small synthetic datasets and compare the subjective interestingness of the approximate solutions against the optimal solutions that we obtain by exhaustive search in these small instances.
Second, we would like to evaluate the efficiency and scalability (how often do we find a cycle, and how fast) of \localSC on real-word datasets.
Finally, we provide two practical use cases. Our Python and Matlab code is publicly available at \url{https://bitbucket.org/ghentdatascience/interesting-cycles-public}. 

\begin{table}[h]
\caption{\small Network statistics.\label{table:stats}}
\resizebox{\columnwidth}{!}{%
\begin{tabular}{lrrl}
\hline
Dataset & |V| & |E| & Edge Weights  \\
\hline
Food web\tablefootnote{\url{http://vlado.fmf.uni-lj.si/pub/networks/data/bio/foodweb/foodweb.htm}} & $128$ & $2\,106$ &{Carbon exchange in the Florida Bay Area}\\
Trade\tablefootnote{\url{https://wits.worldbank.org/}}  & $221$ &$ 1\,957$ &{Country top 5 import \& export in 2018}\\
Gnutella\tablefootnote{\url{https://snap.stanford.edu/data/p2p-Gnutella31.html}} &$26\,518$ & $65\,369$ & {Connections between hosts in a p2p-network}\\
Enron\tablefootnote{\url{https://snap.stanford.edu/data/email-Enron.html}}& $36\,692$ &$ 183\,831$ &{Email communication network}\\
\hline
\end{tabular}%
}
\end{table}


\subsection{Quality experiments on synthetic datasets}
In this section we evaluate the quality of solutions obtained by the algorithms for \msic, and our local Steiner cycle search heuristic \kmsic, using various choices of $\alpha$ and $\beta$. This requires to exhaustively search for their optimal solutions, by enumerating all the cycles using Johnson's algorithm \citep{Johnson} that runs in $\bigO((n+m)(c+1))$ time, where $c$ is the total number of cycles in the input graph. To keep the exhaustive search feasible, we perform the quality tests on small instances and generated $200$ random Erd\H{o}s-R\'{e}nyi graphs with $n = 20$ and edge probability $0.2$. Even in such small instances, we found an average of $218,080$ cycles per instance, with the maximum number of cycles found in an instance being more than $5$ million. 
We set the weight of each edge to a random integer that is generated uniformly at random from the interval $[1, 10K]$. 

We start by evaluating the quality of solutions obtained by Karp's algorithm and its variant for \msic. We use varying values of $\alpha$ and $\beta$ obtained by evaluating (\ref{eq:alphaBeta}) for $q \in \{0.1, 0.2, 0.3\}$. Figure~\ref{fig:expKarp} shows the relative performance w.r.t. the optimal solution for different values of $q$ over $200$ random Erd\H{o}s-R\'{e}nyi instances, sorted from worst to best performing. 
In order to have a baseline, we compute the average interestingness over all possible cycles that were encountered in the 200 instances.
The influence of the parameter $q$ is clearly visible. For $q=0.1$, Karp's algorithm provides the optimal cycle in about $10\%$ of the instances, and has a performance ratio of at least $0.75$ in the rest of the instances. For $q=0.2$, Karp's algorithm provides the optimal cycle only in $5$ instances, with a performance ratio of at least $0.5$ overall. For $q=0.3$, the performance ratio drops drastically as expected, since, the optimal cycle corresponds more to the longest weighted cycle, while Karp's algorithm provides the cycle with maximum mean weight.
Interestingly, the variant from Section~\ref{sec:algo:var} performs slightly worse overall than Karp's algorithm. However, we report that in a small number of instances it performed significantly better than Karp's algorithm although this trend didn't generalize.  
As a guideline, we advise to set $q$ to be not larger than the density of the network (which is 0.2 in this case).

\begin{figure}
\centering
\includegraphics[width=0.35\textwidth]{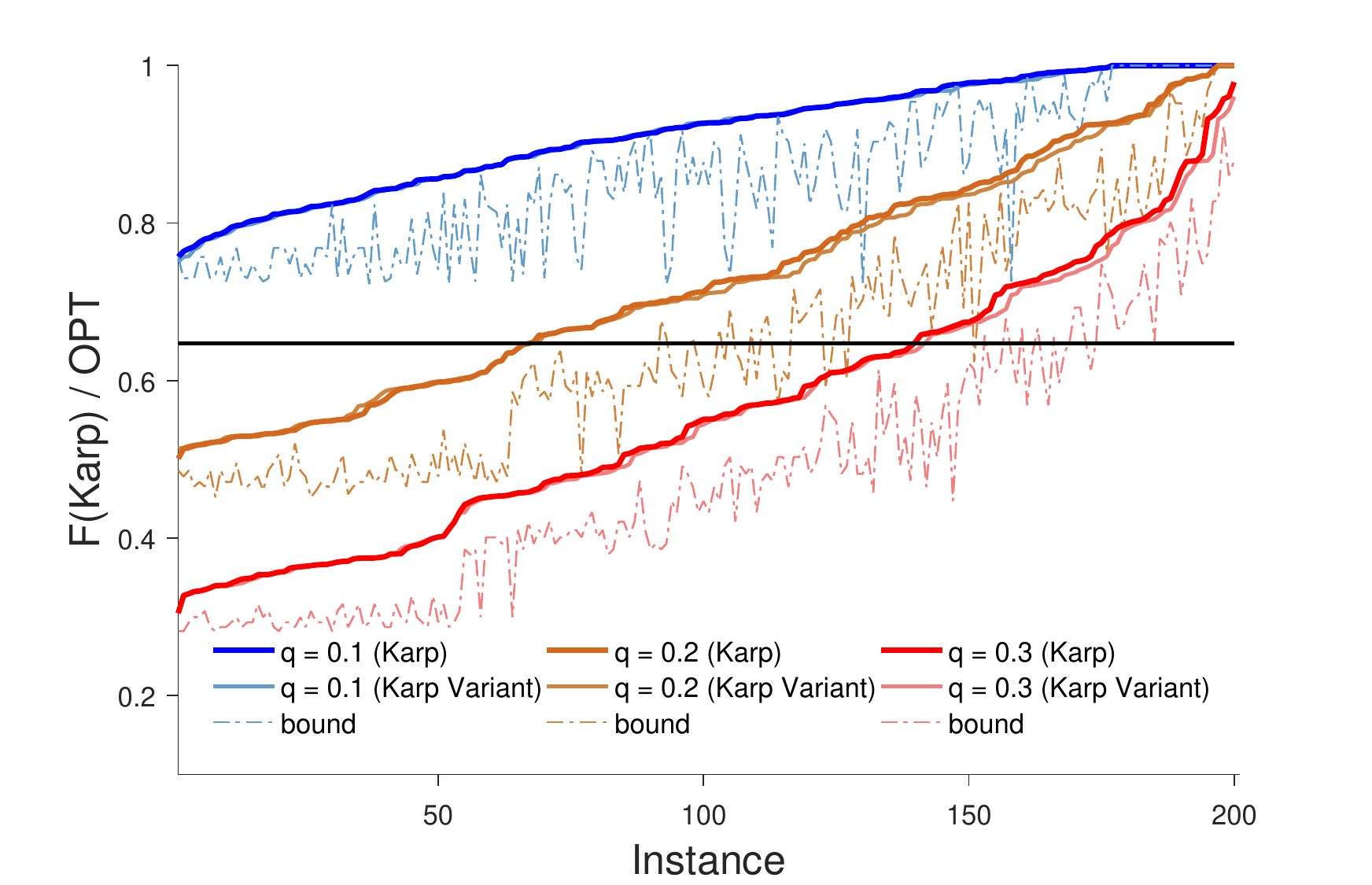} 
\caption{\small Relative performance of Karp's \mmc \& Karp's Variant for various $q$. The dashed lines indicate the theoretical bound provided in Lemma~\ref{karpapp}. \label{fig:expKarp}}
\end{figure}

Next we evaluate the quality of solutions obtained by \localSC for \kmsic. 
We set $q=0.05$ and randomly pick $k$ terminal nodes, for $k \in \{1,5,10\}$. We set no upper bound on the maximum cycle length, i.e., $\lenmax = 20$, run the algorithm $5$ times, and pick the best solution.
Relative performance is shown in Figure~\ref{fig:expSteiner}. Instances in the $x$-axis are again sorted from worse to best performing. The dashed lines indicate the best value of an initial Steiner cycle that was found in the 5 tries, clearly showing that the sequence of changes proposed in Section~\ref{sec:methkmsic} improve the score by a good amount. We also observed that \localSC didn't find any Steiner cycle in $55$ out of $200$ instances for $k = 10$, while this number was $25$ for $k=5$ and $8$ for $k=1$. 
The increase in the performance for larger $k$ is mainly due to the fact that there are more possible local changes available to perform on an initially found cycle for higher $k$, provided that a Steiner cycle of length at most $\lenmax$ exists.

We analyze the running time of \localSC in two different settings, see Figure~\ref{fig:kandn}.
First, we generate Erd\H{o}s graphs of size $n=20$ with edge probability 0.2, set no bound on $\lenmax$, and let the query size $k$ vary. For each $k$, we generate 50 graphs and repeat the algorithm one time.
As expected, for fixed $n$ and $m$, the running time is linear in $k$.
Second, we set $k=3$, $\lenmax=10$ and let the graph size $n$ vary. Again for each $n$, we generate 50 instances.
As expected, there is a polynomial dependence on the graph size $n$; doubling the graph size $n$ roughly leads to a quadrupling in running time.
Karps's MMC and Karp's Variant \emph{always} run in $\Theta(nm)$ time, hence are not tested. Their space complexity is given by $\Theta(n^2)$.

\begin{figure}
\centering
\includegraphics[width=0.35\textwidth]{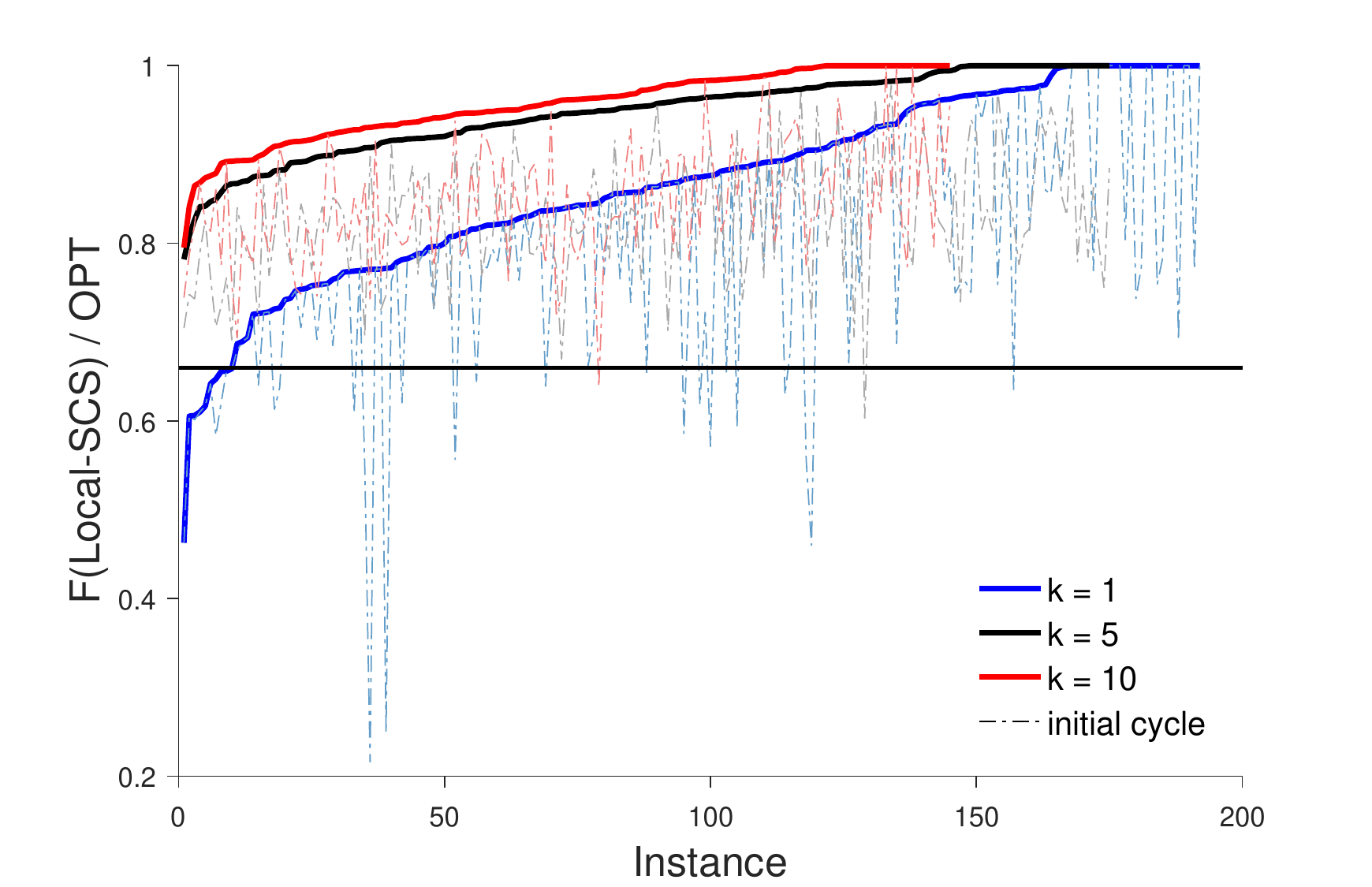}
\caption{\small Relative performance of \localSC for various $k$ and $q=0.05$. The dashed lines indicate the best initial solutions before applying changes.\label{fig:expSteiner}}
\end{figure}

\begin{figure}[tp]
\captionsetup[subfloat]{farskip=5pt,captionskip=1pt,labelformat=empty}
\centering
\subfloat{\includegraphics[width=.47\linewidth]{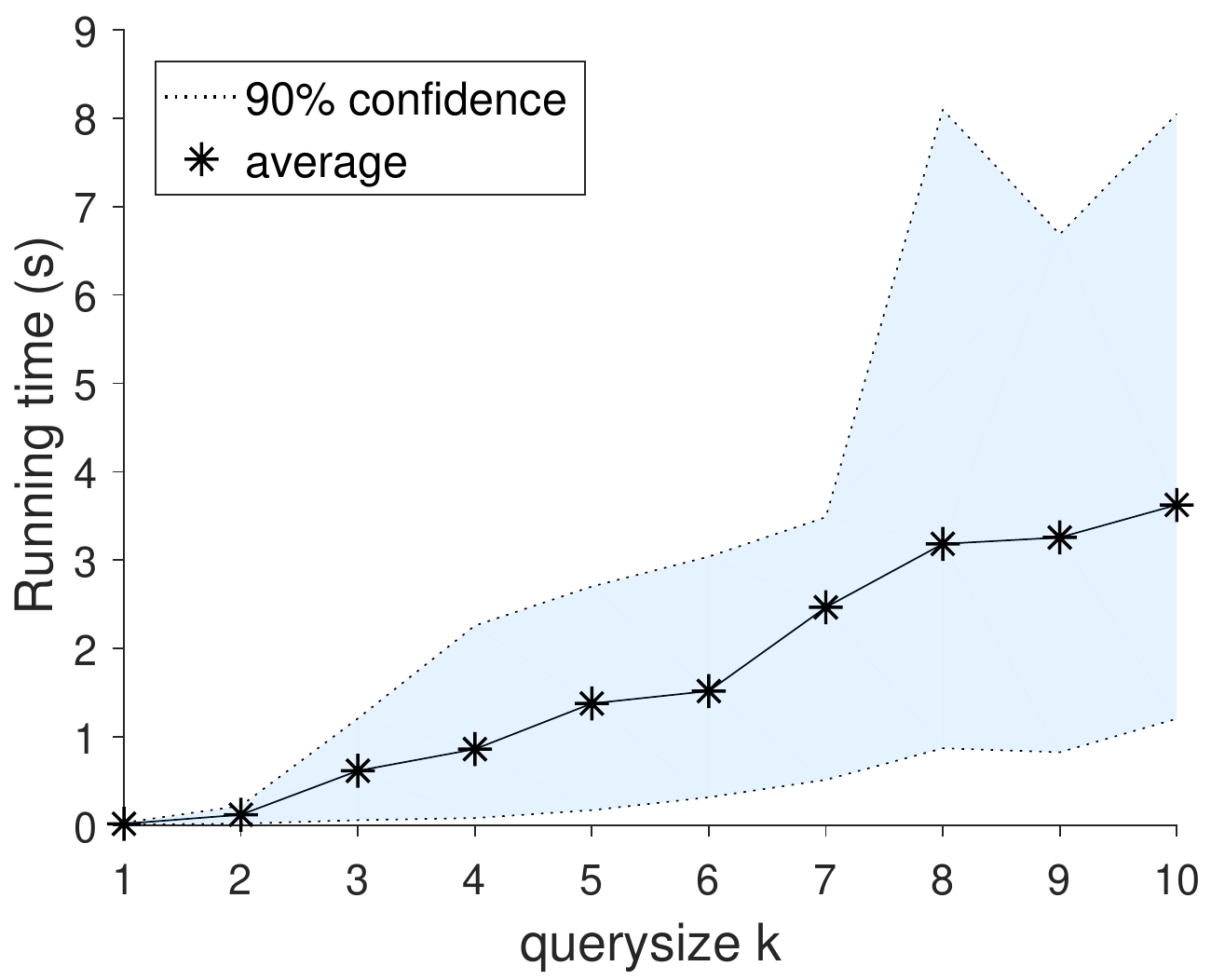}}
\subfloat{\includegraphics[width=.5\linewidth]{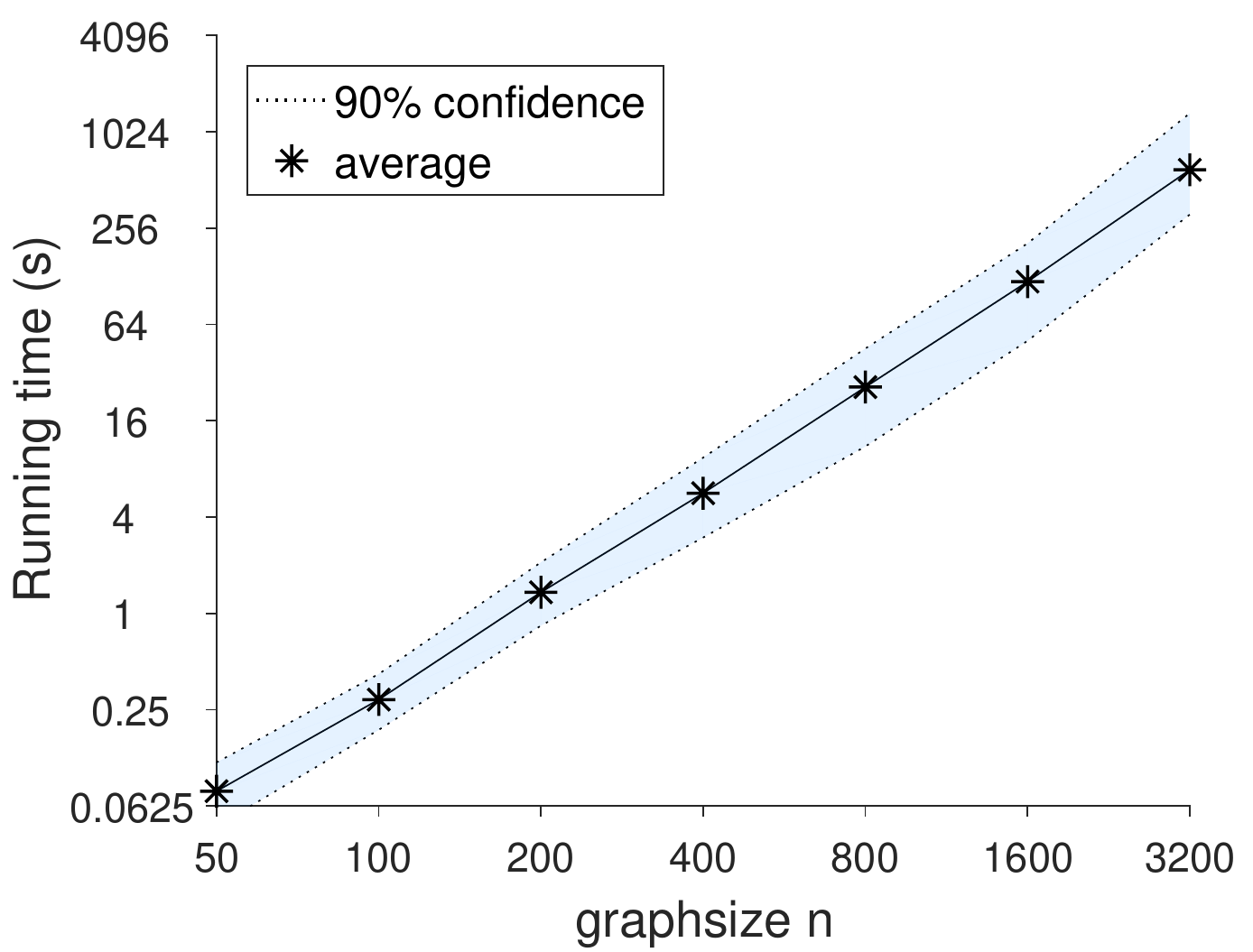}}
\caption{Running times of \localSC on Erd\H{o}s graphs: on the left for varying query size $k$, on the right for varying graph size $n$. \label{fig:kandn}}
\end{figure}

\begin{figure*}
\centering
\includegraphics[width=0.80\textwidth]{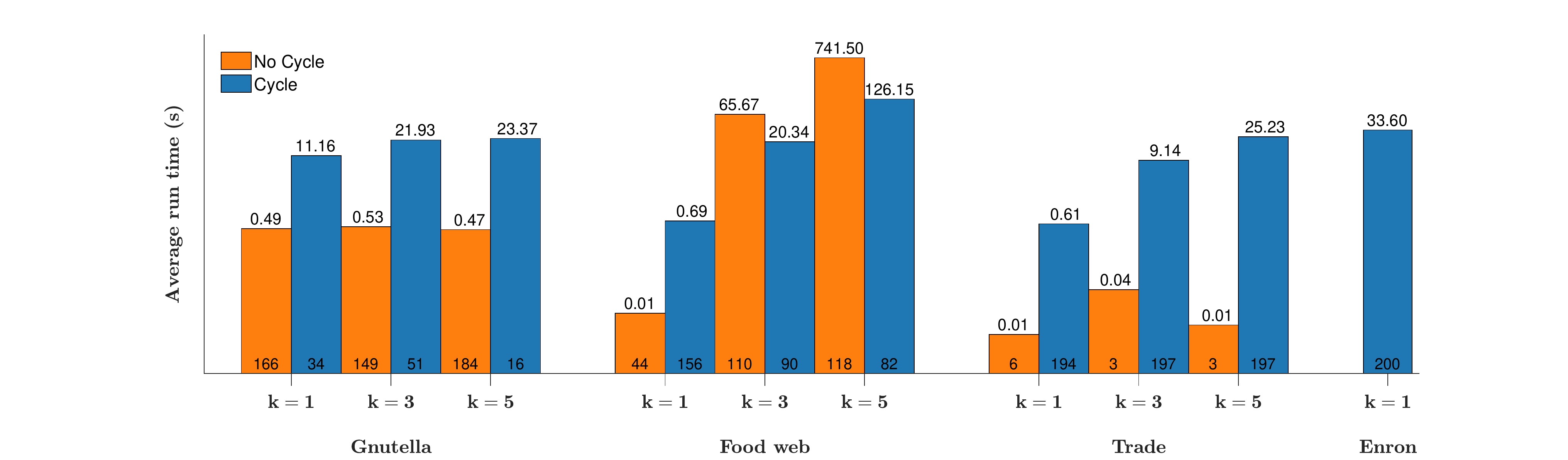}
\caption{\small Average running time of \localSC on real-world datasets.\label{fig:timing}}
\end{figure*} 

\subsection{Scalability on real-world datasets}
Contrary to the dense Erd\H{o}s graphs, it is interesting to test scalability on (sparser) real-world datasets.
Indeed, often it might be the case that we don't find a Steiner cycle at all.
We test on datasets whose basic statistics are summarized in Table~\ref{table:stats}. 
For each dataset and for each $k \in \{1,3,5\}$, we generated a set $Q_i$ of $k$ random terminal nodes, for $i \in [1, 200]$. 
Each $Q_i$ is generated by a snowball sampling scheme, i.e., after choosing an initial terminal node into $Q_i$ at random, each of its neighbors is chosen into $Q_i$ with probability $p$ until $k$ terminals are obtained. We set $p = 0.4$ in our experiments. We set $ \lenmax = 10$ for the local search and we run the algorithm only once for $q=0.01$. Figure~\ref{fig:timing} shows the average running time of \localSC to find an interesting Steiner cycle of length at most $10$ or return none if such a Steiner cycle doesn't exist. The numbers above the $x$-axis (inside the bars) denote how many times a solution was found or not. For Enron, we only report the results for $k=1$ since for $k=3$, the algorithm couldn't terminate within 24 hours for some query sets. 
For Gnutella, $k$ doesn't seem to influence the running time which is expected since the network is sparse with $m$ being roughly $2$ times $n$. Indeed, in most cases \localSC didn't find a solution.
\subsection{Practical use cases}
We discuss two practical use cases for discovering interesting cycles: food trajectories and financial trade data.
\subsubsection{Food Web dataset}
The Florida Bay Food Web dataset \citep{floridabay} contains information about carbon exchange between species in the Florida Bay Area.
The dataset consists of $128$ species and $2106$ edges. An edge weight is a snapshot of the amount of biomass transferred between the species within a fixed period of time. 
Besides a wide variety of organisms, ranging from microorganisms to sharks, it also contains special nodes such as ``input'', ``output'' or ``particulate
organic carbon (POC)''. The extinction of certain species could have a severe impact on the diet of the other organisms in the food chain, possibly leading to a sensitive cascading effect. Thus, in a food web, we consider cycles as interesting based on their vulnerability to extinction. Using the biomass edge weights, we fit a geometric background distribution on the network. Given such a model, an edge $(a,b)$ is usually more informative if the weight of the edge is either a large part of the total incoming weight of node $b$, or a large part of the outgoing weight of node $a$ (or both).

Figure~\ref{fig:foodweb1}(a)-(c) shows the top 3 results after iteratively mining the most interesting cycles using the \mmc algorithm. After a cycle is shown to the user, all edges in the cycle are known to the user, hence, convey zero information. This amounts to setting the information content of those edges equal to zero, after which the algorithm can be applied again.
The top 3 resulting cycles all contain edges that are quite vulnerable to extinction. For e.g., the diet of the ``Filefishes'' consists of 47\% of ``Echinoderma'', a fungus species, which means that a disease in the population of ``Echinoderma' will most likely have a significant impact on the population of ``Filefishes''. This fraction is quite high compared to the expected fraction on an edge being 6\% in the network. 
Note that the top 2 cycles are cannibalism cycles: it contains a predatory subspecies eating its own kind.
Figure~\ref{fig:foodweb1}(d) shows the results when we query a cycle containing both ``Snook'' and ``Crocodile'', for $q=0.1$ and $\lenmax =6$.
This cycle also contains a link between the ``Omnivorous Crab'' and the ``Atlantic Blue Crab'', a highly interesting connection that was also found in Figure~\ref{fig:foodweb1}(a).

\begin{figure*}[t!]
\begin{tabular}{cccc}
\begin{tikzpicture}[node distance =1.75cm, every text node part/.style={align=center}]
\label{fig:fooda}
\def \margin {5}
\def \radius {0.9cm}
\def \rm {0.2cm}

\node[draw, fill, circle, scale = 0.4, label={[xshift=0cm, yshift=0cm]\footnotesize Atlantic \\  \footnotesize Blue Crab}] at (90:\radius) {};
\node[draw, fill, circle, scale = 0.4, label={[xshift=0.8cm, yshift=-.6cm]\footnotesize Omnivorous \\  \footnotesize Crabs}] at (330:\radius) {};
\node[draw, fill, circle, scale = 0.4, label={[xshift=-0.4cm, yshift=-.6cm]\footnotesize Benthic \\ \footnotesize POC}] at (210:\radius) {};

\node[fill, circle, scale = 0, label={[xshift=-0.1cm, yshift=-0.2cm]\footnotesize 67\%}] at (30:\radius-\rm) {};
\node[fill, circle, scale = 0, label={[xshift=0.1cm, yshift=-0.2cm]\footnotesize <1\%}] at (150:\radius-\rm) {};
\node[fill, circle, scale = 0, label={[xshift=0cm, yshift=-0.2cm]\footnotesize 35\%}] at (270:\radius-\rm) {};

\draw[->, >=latex] (-30+\margin:\radius) 
    arc (-30+\margin:90-\margin:\radius);
\draw[->, >=latex] (90+\margin:\radius) 
    arc (90+\margin:210-\margin:\radius);
\draw[->, >=latex] (210+\margin:\radius) 
    arc (210+\margin:330-\margin:\radius);

\end{tikzpicture}
& \hspace{-2mm}
\begin{tikzpicture}[node distance =1cm, every text node part/.style={align=center}]
\label{fig:foodc}
\def \margin {5}
\def \radius {0.9cm}
\def \rm {0.2cm}

\node[draw, fill, circle, scale = 0.4, label={[xshift=0.5cm, yshift=-0.2cm]\footnotesize Benthic \\ \footnotesize POC}] at (45/2:\radius) {};
\node[draw, fill, circle, scale = 0.4, label={[xshift=0.6cm, yshift=-.1cm]\footnotesize Detritivorous \\  \footnotesize Amphipods}] at (45+45/2:\radius) {};
\node[draw, fill, circle, scale = 0.4, label={[xshift=-0.5cm, yshift=0cm]\footnotesize Other  \\ \footnotesize Cnidaridae}] at (2*45+45/2:\radius) {};
\node[draw, fill, circle, scale = 0.4, label={[xshift=-0.8cm, yshift=-0.1cm]\footnotesize Echinoderma}] at (3*45+45/2:\radius) {};
\node[draw, fill, circle, scale = 0.4, label={[xshift=-0.7cm, yshift=-0.2cm]\footnotesize Filefishes}] at (4*45+45/2:\radius) {};
\node[draw, fill, circle, scale = 0.4, label={[xshift=-0.4cm, yshift=-0.5cm]\footnotesize Water POC}] at (5*45+45/2:\radius) {};
\node[draw, fill, circle, scale = 0.4, label={[xshift=0.4cm, yshift=-0.5cm]\footnotesize Bivalves}] at (6*45+45/2:\radius) {};
\node[draw, fill, circle, scale = 0.4, label={[xshift=0.4cm, yshift=-0.2cm]\footnotesize Rays}] at (7*45+45/2:\radius) {};

\node[fill, circle, scale = 0, label={[xshift=-0.175cm, yshift=-0.3cm]\footnotesize 27\%}] at (45:\radius) {};
\node[fill, circle, scale = 0, label={[xshift=0cm, yshift=-0.4cm]\footnotesize 23\%}] at (90:\radius) {};
\node[fill, circle, scale = 0, label={[xshift=0.2cm, yshift=-0.3cm]\footnotesize 6\%}] at (3*45:\radius) {};
\node[fill, circle, scale = 0, label={[xshift=0.3cm, yshift=-0.15cm]\footnotesize 47\%}] at (4*45:\radius) {};
\node[fill, circle, scale = 0, label={[xshift=0.21cm, yshift=0cm]\footnotesize <1\%}] at (5*45:\radius) {};
\node[fill, circle, scale = 0, label={[xshift=0cm, yshift=0cm]\footnotesize 20\%}] at (6*45:\radius) {};
\node[fill, circle, scale = 0, label={[xshift=-0.14cm, yshift=0cm]\footnotesize 37\%}] at (7*45:\radius) {};
\node[fill, circle, scale = 0, label={[xshift=-0.3cm, yshift=-0.2cm]\footnotesize <1\%}] at (8*45:\radius) {};

\draw[->, >=latex] (45/2+\margin:\radius) 
    arc (45/2+\margin:45+45/2-\margin:\radius);
\draw[->, >=latex] (45+45/2+\margin:\radius) 
    arc (45+45/2+\margin:2*45+45/2-\margin:\radius);
\draw[->, >=latex] (2*45+45/2+\margin:\radius) 
    arc (2*45+45/2+\margin:3*45+45/2-\margin:\radius);
\draw[->, >=latex] (3*45+45/2+\margin:\radius) 
    arc (3*45+45/2+\margin:4*45+45/2-\margin:\radius);
\draw[->, >=latex] (4*45+45/2+\margin:\radius) 
    arc (4*45+45/2+\margin:5*45+45/2-\margin:\radius);
\draw[->, >=latex] (5*45+45/2+\margin:\radius) 
    arc (5*45+45/2+\margin:6*45+45/2-\margin:\radius);
\draw[->, >=latex] (6*45+45/2+\margin:\radius) 
    arc (6*45+45/2+\margin:7*45+45/2-\margin:\radius);
\draw[->, >=latex] (7*45+45/2+\margin:\radius) 
    arc (7*45+45/2+\margin:8*45+45/2-\margin:\radius);

\end{tikzpicture}
& \hspace{-1mm}
\begin{tikzpicture}[node distance =1.75cm, every text node part/.style={align=center}]
\label{fig:foodb}
\def \margin {5}
\def \radius {0.9cm}
\def \rm {0.2cm}

\node[draw, fill, circle, scale = 0.4, label={[xshift=0cm, yshift=0cm]\footnotesize Predatory \\  \footnotesize Gastropods}] at (90:\radius) {};
\node[draw, fill, circle, scale = 0.4, label={[xshift=0.8cm, yshift=-.6cm]\footnotesize Detritivorous \\  \footnotesize Gastropods}] at (330:\radius) {};
\node[draw, fill, circle, scale = 0.4, label={[xshift=-0.4cm, yshift=-.6cm]\footnotesize Benthic \\ \footnotesize POC}] at (210:\radius) {};

\node[fill, circle, scale = 0, label={[xshift=-0.1cm, yshift=-0.2cm]\footnotesize 52\%}] at (30:\radius-\rm) {};
\node[fill, circle, scale = 0, label={[xshift=0.1cm, yshift=-0.2cm]\footnotesize <1\%}] at (150:\radius-\rm) {};
\node[fill, circle, scale = 0, label={[xshift=0cm, yshift=-0.2cm]\footnotesize 47\%}] at (270:\radius-\rm) {};

\draw[->, >=latex] (-30+\margin:\radius) 
    arc (-30+\margin:90-\margin:\radius);
\draw[->, >=latex] (90+\margin:\radius) 
    arc (90+\margin:210-\margin:\radius);
\draw[->, >=latex] (210+\margin:\radius) 
    arc (210+\margin:330-\margin:\radius);

\end{tikzpicture}
& \hspace{-2mm}
\begin{tikzpicture}[node distance =1cm, every text node part/.style={align=center}]
\label{fig:foodd}
\def \margin {5}
\def \radius {0.9cm}
\def \rm {0.2cm}

\node[draw, fill, circle, scale = 0.4, label={[xshift=0.55cm, yshift=-0.2cm]\footnotesize \underline{\textbf{Snook}}}] at (30:\radius) {};
\node[draw, fill, circle, scale = 0.4, label={[xshift=0cm, yshift=0cm]\footnotesize \underline{\textbf{Crocodile}}}] at (90:\radius) {};
\node[draw, fill, circle, scale = 0.4, label={[xshift=-0.7cm, yshift=-0.1cm]\footnotesize Water POC}] at (150:\radius) {};
\node[draw, fill, circle, scale = 0.4, label={[xshift=-0.6cm, yshift=-0.3cm]\footnotesize Bivalves}] at (210:\radius) {};
\node[draw, fill, circle, scale = 0.4, label={[xshift=0cm, yshift=-0.55cm]\footnotesize Omnivorous Crabs}] at (270:\radius) {};
\node[draw, fill, circle, scale = 0.4, label={[xshift=0.7cm, yshift=-0.5cm]\footnotesize Atlantic \\ \footnotesize Blue Crab}] at (330:\radius) {};

\node[fill, circle, scale = 0, label={[xshift=-0.25cm, yshift=-0.3cm]\footnotesize 3\%}] at (0:\radius) {};
\node[fill, circle, scale = 0, label={[xshift=0cm, yshift=-0.4cm]\footnotesize <1\%}] at (60:\radius) {};
\node[fill, circle, scale = 0, label={[xshift=0.3cm, yshift=-0.3cm]\footnotesize <1\%}] at (120:\radius) {};
\node[fill, circle, scale = 0, label={[xshift=0.3cm, yshift=-0.15cm]\footnotesize 20\%}] at (180:\radius) {};
\node[fill, circle, scale = 0, label={[xshift=0.16cm, yshift=0cm]\footnotesize 5\%}] at (240:\radius) {};
\node[fill, circle, scale = 0, label={[xshift=-.1cm, yshift=0cm]\footnotesize 67\%}] at (300:\radius) {};

\draw[->, >=latex] (30+\margin:\radius) 
    arc (30+\margin:90-\margin:\radius);
\draw[->, >=latex] (90+\margin:\radius) 
    arc (90+\margin:150-\margin:\radius);
\draw[->, >=latex] (150+\margin:\radius) 
    arc (150+\margin:210-\margin:\radius);
\draw[->, >=latex] (210+\margin:\radius) 
    arc (210+\margin:270-\margin:\radius);
\draw[->, >=latex] (270+\margin:\radius) 
    arc (270+\margin:330-\margin:\radius);
\draw[->, >=latex] (-30+\margin:\radius) 
    arc (-30+\margin:30-\margin:\radius);

\end{tikzpicture}
\\ 
(a) & (b) & (c) & (d)    \\
\end{tabular}
\caption{\small (a)-(c): The top 3 results from iteratively mining the most interesting cycles in the Food Web dataset. The percentages next to each edge $(u,v)$ indicates the fraction of the total incoming weight of species $v$. (d) the result when querying ``Snook'' and ``Crocodile''. \label{fig:foodweb1}}
\end{figure*}
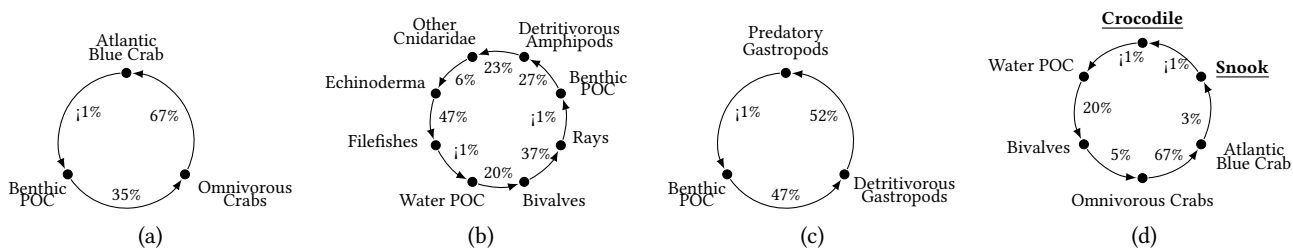

\subsubsection{Trade dataset}
To see the influence of a prior belief model on the resulting cycles, we look at the trading volumes between countries in $2018$.
We set $\lenmax=6$, $q=0.01$, and used $10$ iterations of \localSC. 
First, we fit a geometric model with the weighted in- and outdegree of each node as a prior.
Figure~\ref{fig:tradea} shows the most interesting cycle in the graph: a 2-cycle between the U.S. and the Dominican Republic.
As discussed in Section~\ref{sec:background}, these edges are indeed very interesting: the Dominican Republic is extremely economically-dependent on the U.S. in terms of import and export. However, the converse is not true. 

Figure~\ref{fig:tradeb} shows the most interesting cycle when we take the bilateral trade agreement between the U.S. and Dominican Republic into account as a prior belief.
Since these edges are now more expected, they become less interesting. The new most interesting cycle is another 2-cycle, between China and Sudan. Again, a small country that is economically-dependent on a bigger country.
Figure~\ref{fig:tradec} shows the result when we query both Iran and the U.S., two countries not expected to be in a direct trade relationship because of the U.S. trade embargo on Iran.
This cycle now contains China as an export country for Iran and China linking back to the U.S. 
Figure~\ref{fig:traded} shows the result when we take the trade relationships between the U.S. and China into account as well.
The direct edge is now expected, and the resulting heuristic takes this into account by placing an intermediate country in between, Nicaragua.
Nicaragua heavily depends on China for its import, and the U.S. for its export, thus making these connections interesting.

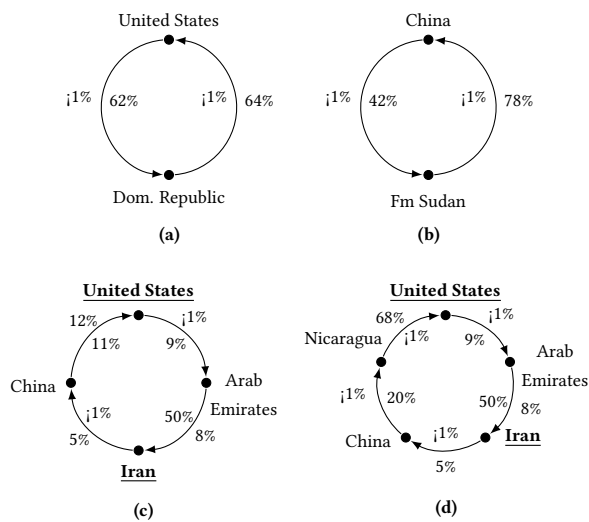
\begin{figure}
\centering
\captionsetup[subfigure]{}
\subfloat[]{
\begin{tikzpicture}[node distance =1.75cm, every text node part/.style={align=center}]
\label{fig:tradea}
\def \margin {5}
\def \radius {0.9cm}
\def \rm {0.2cm}

\node[draw, fill, circle, scale = 0.4, label={[xshift=0cm, yshift=0cm]\footnotesize United States}] at (90:\radius) {};
\node[draw, fill, circle, scale = 0.4, label={[xshift=0cm, yshift=-.6cm]\footnotesize Dom. Republic}] at (270:\radius) {};

\node[fill, circle, scale = 0, label={[xshift=-0.3cm, yshift=-0.1cm]\footnotesize <1\%}] at (0:\radius) {};
\node[fill, circle, scale = 0, label={[xshift=0.3cm, yshift=-0.1cm]\footnotesize 64\%}] at (0:\radius) {};

\node[fill, circle, scale = 0, label={[xshift=-0.3cm, yshift=-0.1cm]\footnotesize <1\%}] at (180:\radius) {};
\node[fill, circle, scale = 0, label={[xshift=0.3cm, yshift=-0.1cm]\footnotesize 62\%}] at (180:\radius) {};

\draw[->, >=latex] (-90+\margin:\radius) 
	arc (-90+\margin:90-\margin:\radius);

\draw[->, >=latex] (90+\margin:\radius) 
	arc (90+\margin:270-\margin:\radius);

\end{tikzpicture}
}\quad
\subfloat[]{
\begin{tikzpicture}[node distance =1.75cm, every text node part/.style={align=center}]
\label{fig:tradeb}
\def \margin {5}
\def \radius {0.9cm}
\def \rm {0.2cm}

\node[draw, fill, circle, scale = 0.4, label={[xshift=0cm, yshift=0cm]\footnotesize China}] at (90:\radius) {};
\node[draw, fill, circle, scale = 0.4, label={[xshift=0cm, yshift=-.6cm]\footnotesize Fm Sudan}] at (270:\radius) {};

\node[fill, circle, scale = 0, label={[xshift=-0.3cm, yshift=-0.1cm]\footnotesize <1\%}] at (0:\radius) {};
\node[fill, circle, scale = 0, label={[xshift=0.3cm, yshift=-0.1cm]\footnotesize 78\%}] at (0:\radius) {};

\node[fill, circle, scale = 0, label={[xshift=-0.3cm, yshift=-0.1cm]\footnotesize <1\%}] at (180:\radius) {};
\node[fill, circle, scale = 0, label={[xshift=0.3cm, yshift=-0.1cm]\footnotesize 42\%}] at (180:\radius) {};

\draw[->, >=latex] (-90+\margin:\radius) 
	arc (-90+\margin:90-\margin:\radius);

\draw[->, >=latex] (90+\margin:\radius) 
	arc (90+\margin:270-\margin:\radius);

\end{tikzpicture}
}

\subfloat[]{
\begin{tikzpicture}[node distance =1.75cm, every text node part/.style={align=center},baseline]
\label{fig:tradec}
\def \margin {5}
\def \radius {0.9cm}
\def \rm {0.2cm}

\node[draw, fill, circle, scale = 0.4, label={[xshift=0cm, yshift=0cm]\footnotesize \underline{\textbf{United States}}}] at (90:\radius) {};
\node[draw, fill, circle, scale = 0.4, label={[xshift=0.5cm, yshift=-.6cm]\footnotesize Arab \\ \footnotesize Emirates}] at (0:\radius) {};
\node[draw, fill, circle, scale = 0.4, label={[xshift=-0.5cm, yshift=-.3cm]\footnotesize China}] at (180:\radius) {};
\node[draw, fill, circle, scale = 0.4, label={[xshift=0cm, yshift=-.6cm]\footnotesize \underline{\textbf{Iran}}}] at (270:\radius) {};

\node[fill, circle, scale = 0, label={[xshift=0.1cm, yshift=0cm]\footnotesize <1\%}] at (45:\radius) {};
\node[fill, circle, scale = 0, label={[xshift=-0.14cm, yshift=-0.3cm]\footnotesize 9\%}] at (45:\radius) {};

\node[fill, circle, scale = 0, label={[xshift=-0.1cm, yshift=0cm]\footnotesize 12\%}] at (135:\radius) {};
\node[fill, circle, scale = 0, label={[xshift=0.2cm, yshift=-0.3cm]\footnotesize 11\%}] at (135:\radius) {};

\node[fill, circle, scale = 0, label={[xshift=0.1cm, yshift=0cm]\footnotesize <1\%}] at (225:\radius) {};
\node[fill, circle, scale = 0, label={[xshift=-0.15cm, yshift=-0.3cm]\footnotesize 5\%}] at (225:\radius) {};

\node[fill, circle, scale = 0, label={[xshift=-0.1cm, yshift=0cm]\footnotesize 50\%}] at (315:\radius) {};
\node[fill, circle, scale = 0, label={[xshift=0.25cm, yshift=-0.25cm]\footnotesize 8\%}] at (315:\radius) {};

\draw[<-, >=latex] (0+\margin:\radius) 
	arc (0+\margin:90-\margin:\radius);

\draw[<-, >=latex] (90+\margin:\radius) 
	arc (90+\margin:180-\margin:\radius);

\draw[<-, >=latex] (180+\margin:\radius) 
	arc (180+\margin:270-\margin:\radius);

\draw[<-, >=latex] (270+\margin:\radius) 
	arc (270+\margin:360-\margin:\radius);

\end{tikzpicture}
}
\subfloat[]{
\begin{tikzpicture}[node distance =1.75cm, every text node part/.style={align=center},baseline]
\label{fig:traded}
\def \margin {5}
\def \radius {0.9cm}
\def \rm {0.2cm}

\node[draw, fill, circle, scale = 0.4, label={[xshift=0cm, yshift=0cm]\footnotesize \underline{\textbf{United States}}}] at (90:\radius) {};
\node[draw, fill, circle, scale = 0.4, label={[xshift=0.6cm, yshift=-0.5cm]\footnotesize Arab\\ \footnotesize Emirates}] at (90-72:\radius) {};
\node[draw, fill, circle, scale = 0.4, label={[xshift=-0.5cm, yshift=0cm]\footnotesize Nicaragua}] at (90+72:\radius) {};
\node[draw, fill, circle, scale = 0.4, label={[xshift=-0.5cm, yshift=-0.3cm]\footnotesize China}] at (90+2*72:\radius) {};
\node[draw, fill, circle, scale = 0.4, label={[xshift=0.5cm, yshift=-0.3cm]\footnotesize \underline{\textbf{Iran}}}] at (90+3*72:\radius) {};

\node[draw, fill, circle, scale = 0, label={[xshift=0cm, yshift=-.5cm]\footnotesize }] at (270:\radius) {};

\node[fill, circle, scale = 0, label={[xshift=0.1cm, yshift=0.05cm]\footnotesize <1\%}] at (45:\radius) {};
\node[fill, circle, scale = 0, label={[xshift=-0.25cm, yshift=-0.25cm]\footnotesize 9\%}] at (45:\radius) {};
\node[fill, circle, scale = 0, label={[xshift=-0.1cm, yshift=0.05cm]\footnotesize 68\%}] at (135:\radius) {};
\node[fill, circle, scale = 0, label={[xshift=0.25cm, yshift=-0.25cm]\footnotesize <1\%}] at (135:\radius) {};

\node[fill, circle, scale = 0, label={[xshift=-0.6cm, yshift=0.25cm]\footnotesize <1\%}] at (225:\radius) {};
\node[fill, circle, scale = 0, label={[xshift=0.05cm, yshift=0.25cm]\footnotesize 20\%}] at (225:\radius) {};

\node[fill, circle, scale = 0, label={[xshift=0cm, yshift=-0.4cm]\footnotesize 5\%}] at (270:\radius) {};
\node[fill, circle, scale = 0, label={[xshift=0cm, yshift=0cm]\footnotesize <1\%}] at (270:\radius) {};

\node[fill, circle, scale = 0, label={[xshift=0cm, yshift=0.2cm]\footnotesize 50\%}] at (-45:\radius) {};
\node[fill, circle, scale = 0, label={[xshift=0.5cm, yshift=0.1cm]\footnotesize 8\%}] at (-45:\radius) {};

\draw[<-, >=latex] (90-72+\margin:\radius) 
	arc (90-72+\margin:90-\margin:\radius);

\draw[<-, >=latex] (90+\margin:\radius) 
	arc (90+\margin:90+72-\margin:\radius);

\draw[<-, >=latex] (90+72+\margin:\radius) 
	arc (90+72+\margin:90+2*72-\margin:\radius);

\draw[<-, >=latex] (90+2*72+\margin:\radius) 
	arc (90+2*72+\margin:90+3*72-\margin:\radius);

\draw[<-, >=latex] (90+3*72+\margin:\radius) 
	arc (90+3*72+\margin:90+4*72-\margin:\radius);

\end{tikzpicture}
}
\caption{\small The $\%$ outside the circle denotes the weight of an edge $(u,v)$, relative to the total export of $u$.
The $\%$ inside the circle the denotes the weight relative to the total import of $v$.
The most interesting cycles: (a) with a prior on weighted in- and out-degree of each country (b) with a prior on the trading volume between US and Dom. Rep. (c) with Iran and US as query nodes, with a prior on weighted in- and out-degrees of each country (d) Iran and US as query nodes, with a prior on trading volume between US and China.\label{fig:tradecase}}
\end{figure}

\section{Related work}
\label{sec:related}
Discovering cyclic patterns in graphs has not received much attention in the data-mining community.
\citet{Giscard2017EvaluatingBO} evaluate the balance of a signed social network by finding simple cycles.
\citet{Kumar2017FindingST} propose an algorithm for enumerating all simple cycles in a directed temporal network, by extending Johnson's algorithm~\cite{Johnson} to a temporal setting. 
Building on the subjective interestingness framework, \citet{vanleeuwen2016} studied the problem of subjectively interesting dense subgraphs, and \citet{Adriaens2019} studied subjectively interesting Steiner trees and forests. We discuss algorithmically relevant work in Section~\ref{sec:algorithms}.


\section{Conclusions \& Further Challenges}
\label{sec:conclusions}
We introduced the problem of discovering interesting cycles in directed graphs, by formally defining the problem of finding the maximum subjectively interesting (Steiner) cycle. Our work opens interesting directions for future research. First, it is worth to consider the usefulness of a non-simple cycle (a so-called \emph{tour}) as a data-mining pattern. Second, extending our results to undirected graphs is non-trivial. Karp's algorithm does not apply, and more general algorithms for minimum ratio problems have to be considered.





\begin{acks}
This work was supported by
the ERC under the
EU's Seventh Framework Programme (FP7/2007-2013) / ERC
Grant Agreement no.\ 615517,
FWO (project no.\ G091017N, G0F9816N),
the EU's Horizon 2020 research and innovation programme and the FWO under the Marie Sklodowska-Curie Grant Agreement no.\ 665501,
three Academy of Finland projects  (286211, 313927, and 317085), 
and the EC H2020 RIA project ``SoBigData'' (654024).
\end{acks}

%
\bibliographystyle{ACM-Reference-Format}
\bibliography{paper}
%

\end{document}